\documentclass[10pt,superscriptaddress,aps,pra,twocolumn,showpacs,nofootinbib,longbibliography, floatfix]{revtex4-2}

\usepackage{silence}
\usepackage[T1]{fontenc}

\usepackage{amsmath,amssymb,amsthm,amstext,amsfonts}
\usepackage{mathrsfs}
\usepackage{physics}     

\usepackage{braket}

\usepackage{latexsym}
\usepackage{gensymb}
\usepackage{color}
\usepackage{adjustbox}
\usepackage{soul}
\usepackage[normalem]{ulem}

\usepackage{graphicx}
\usepackage{epstopdf}
\usepackage{subcaption}
\usepackage{lipsum}

\usepackage[colorlinks=true, citecolor=blue, urlcolor=blue]{hyperref}

\newcommand{\be}{\begin{equation}}
\newcommand{\ee}{\end{equation}}
\newcommand{\ba}{\begin{eqnarray}}
\newcommand{\ea}{\end{eqnarray}}

\newtheorem{lemma}{Lemma}

\begin{document}

\title{Secure One-Sided Device-Independent Quantum Key Distribution Under Collective Attacks with Enhanced Robustness}

\author{Pritam Roy}
\email{roy.pritamphy@gmail.com}
\affiliation{S. N. Bose National Centre for Basic Sciences, Block JD, Sector III, Salt Lake, Kolkata 700 106, India}

\author{Subhankar Bera}
\email{berasanu007@gmail.com}
\affiliation{S. N. Bose National Centre for Basic Sciences, Block JD, Sector III, Salt Lake, Kolkata 700 106, India}

\author{A. S. Majumdar}
\email{archan@bose.res.in}
\affiliation{S. N. Bose National Centre for Basic Sciences, Block JD, Sector III, Salt Lake, Kolkata 700 106, India}

\begin{abstract}
We study the security of a quantum key distribution (QKD) protocol under the one-sided device-independent (1sDI) setting, which assumes trust in only one party’s measurement device. This approach effectively provides a balance between the experimental viability of device-dependent (DD-QKD) and the minimal trust assumptions of device-independent (DI-QKD).  An analytical lower bound on the asymptotic key rate is derived to provide security against collective attacks, in which the eavesdropper's information is limited only by the function of observed violation of a linear quantum steering inequality, specifically the three-setting Cavalcanti–Jones–Wiseman–Reid (CJWR) inequality. We provide a closed-form key rate formula by reducing the security analysis to mixtures of Bell-diagonal states by utilizing symmetries of the steering functional. We show that the protocol tolerates higher quantum bit error rates (QBER) than present DI-QKD protocols by benchmarking its performance under depolarizing noise. Furthermore, we explore the impact of detection inefficiencies and show that, in contrast to DI-QKD, which requires near-perfect detection, secure key generation can be achieved even with lower detection efficiency on the untrusted side. These findings highlight the advantages of 1sDI-QKD as a steering-based alternative for secure quantum communication and provide insights relevant for near-future experimental implementations.

\end{abstract}

\maketitle

\section{Introduction} 
Quantum key distribution (QKD) allows two parties to share a secret key, with security provided by the principles of quantum physics rather than computational assumptions~\cite{Shannon1949, RSA1978}. The seminal BB84 protocol~\cite{BB84} demonstrated that quantum states cannot be measured without disturbing them, enabling the detection of any eavesdropping attempt. Thereafter, the E91 protocol~\cite{Ekert1991} introduced an entanglement-based approach where security is certified via Bell inequality violation\cite{Bell1964, Clauser1969}, and the BBM92 protocol~\cite{BBM92} proposed a related scheme that employs entanglement\cite{HorodeckiQentanglement2009} without relying on nonlocal correlations. Although these protocols are theoretically secure under idealized assumptions, real-world implementations involve imperfect and potentially untrusted devices, opening up vulnerabilities through various adversarial attacks ~\cite{Lydersen2010, Scarani2009}.

Attack strategies in QKD are typically classified into individual, collective, and coherent attacks, in increasing order of generality. In individual attacks, the adversary measures each signal independently~\cite{Ekert1994, Slutsky1998, LutkenhausIndAttack2000, Curty2005, roy2024sequential}, whereas in collective attacks, Eve interacts identically with each signal but defers measurement for joint processing~\cite{Acin2007, Scarani2009}. Coherent attacks are the most powerful, allowing arbitrary joint operations on all signals~\cite{Shor2000, Masanes2011}. Security proofs against these strategies often rely on entanglement~\cite{HorodeckiQentanglement2009} or Bell nonlocality~\cite{Bell1964, Clauser1969, Acin2006}, and were initially device-dependent~\cite{LoChau1999, Shor2000}. To address trust issues in devices, the device-independent QKD (DI-QKD) paradigm has gained prominence, particularly after foundational security results~\cite{Acin2007, Vazirani14}. DI-QKD has since advanced through protocols using asymmetric Bell inequalities~\cite{Woodhead2021}, random states~\cite{BeraRandomQKD2023}, or random measurement bases~\cite{SchwonnekRandomBasis2021}, and experimental demonstrations~\cite{LiuDI-QKDExpt, Zhang2022DIQKD}. Security is guaranteed solely by the violation of Bell inequalities~\cite{Bell1964, Clauser1969}, making DI-QKD the most robust cryptographic framework. However, its implementation remains challenging due to strict requirements such as high detection efficiency~\cite{Pironio2009} and loophole-free Bell tests~\cite{Shalm2015Loopholefreebelltest, Giustina2015Loopholefreebelltest, Li2018LoopholeFreebelltest, HoNoisyPreProcessQKD2020, Zapatero2023}.

To mitigate the practical limitations of DI-QKD, particularly the need for high detection efficiencies and entanglement, several intermediate frameworks have been introduced. Semi-device-independent QKD (SDI-QKD)\cite{PawlowskiSemiDIQKD2011} ensures security by assuming trusted state preparation while leaving measurement devices uncharacterized. Notably, variants based on quantum contextuality\cite{GuptaSemiDiQKD2023} have demonstrated security without relying on nonlocality. In comparison, measurement-device-independent QKD (MDI-QKD)~\cite{Lo2012MDIQKD, YujunMDI2016} achieves security under the assumption of trusted entangled state preparation, even with untrusted measurement devices.

Under ideal collective attacks, control over just one measurement device and the source is sufficient to compromise security, as demonstrated by the security analyses in Refs.~\cite{Acin2007, Pironio2009}. The 1sDI-QKD framework, which fits in between DI-QKD and SDI/MDI/DD-QKD in the hierarchy of trust models, is based on this observation and assumes trust in only one party's device, usually Bob.
The structure of 1sDI-QKD is consistent with the concept of quantum steering, which was first proposed by Schrödinger~\cite{Schrödinger1935}, and  formalized later by Wiseman \textit{et al.}~\cite{WisemanSteering2007}. Various criteria, such as Reid's uncertainty-based test~\cite{Reid1989}, the CJWR inequality~\cite{CavalcantiExptCriteriaSteering2009, ChenCJWR2013}, entropic~\cite{Schneeloch2013}, fine-grained~\cite{Pramanik2014}, and sum-uncertainty-relation-based formulations~\cite{sumuncert}, can
be used to identify steering, which captures the ability of a trusted party to nonlocally affect the state of an untrusted party. Its distinction from entanglement and Bell nonlocality has been established both theoretically~\cite{Jones2007Steering, Walborn2011} and experimentally~\cite{CavalcantiExptCriteriaSteering2009, Saunders2010EPRSteering}, and further quantified using dedicated measures~\cite{CostaQuantificationSteering2016, UolaQuantumsteering2020, CavalcantiSteeringReview2017, JebaSteeringCost}. 

Several protocols have explored the 1sDI-QKD regime under various assumptions and models~\cite{Branciard2012, Tomamichel1sDIQKD2013, Pramanik2014, Mukherjee2023SteeringQKD, masini20241sDIQKD}. The works of Branciard et al.~\cite{Branciard2012} and Tomamichel et al.~\cite{Tomamichel1sDIQKD2013} have focused on entropy-based security proofs for BBM92-like or prepare-and-measure schemes, and although they align with the steering scenario, their security bounds depend solely on QBER and are not explicitly connected to steering inequality violations. Pramanik et al.~\cite{Pramanik2014} considered individual attacks and demonstrated steering-based security only in that limited regime. Mukherjee et al.~\cite{Mukherjee2023SteeringQKD} focused on the usefulness of steerable states in QKD but did not analyze explicit attack models. More recently, Masini and Sarkar~\cite{masini20241sDIQKD} have employed a semidefinite-programming-based approach for proving security against coherent attacks, but without deriving closed-form expressions.

To the best of our knowledge, no existing 1sDI-QKD protocol derives a closed-form asymptotic key rate that quantitatively depends on the violation of a steering inequality, in direct analogy with how DI-QKD protocols relate Bell violations to Eve's information~\cite{Acin2007, Pironio2009}. This gap motivates the need for analytical key rate expressions based on observable steering violations, which would simplify certification and enhance practical relevance.

The DI-QKD security proof by Acín et al.\cite{Acin2007} is notable for analytically linking Bell inequality violations to asymptotic key rates, enabling device-independent certification based on observable quantities. Motivated by this, in the current work we establish a closed-form bound for 1sDI-QKD where the key rate is directly expressed in terms of steering inequality violation, thus operationalizing the role of steering in secure key generation. Among various criteria~\cite{Reid1989, CavalcantiExptCriteriaSteering2009, Walborn2011, Schneeloch2013, Pramanik2014}, the Cavalcanti–Jones–Wiseman–Reid (CJWR) inequality~\cite{CavalcantiExptCriteriaSteering2009, ChenCJWR2013} is especially suited for this task due to its linearity, geometric clarity, and analytical applicability to a broad class of two-qubit states. 

In this work, we evaluate a 1sDI-QKD protocol that employs the CJWR steering inequality as a security witness. The protocol uses three binary Pauli measurements per party, with key bits extracted from rounds where both parties measure along \(\sigma_z\), ensuring low data leakage. The other rounds are used to estimate the CJWR parameter \(\mathcal{F}_3\), enabling security verification without basis reconciliation. The security of our protocol is analyzed under collective attacks, where the adversary applies identical and independent operations across rounds. A composable lower bound on the asymptotic key rate is derived using the Devetak–Winter formula, with the CJWR violation acting as the key security witness. Leveraging dimensionality reduction and symmetry arguments, we focus on Bell-diagonal states, for which both the Holevo quantity and the CJWR parameter admit closed-form expressions. This enables an explicit key rate formula in terms of the observed QBER and steering inequality violation.

We evaluate the robustness of our CJWR-based 1sDI-QKD protocol under depolarizing noise, showing that it tolerates a QBER of up to 8.62\%, higher than standard DI-QKD protocols, while relying on weaker trust assumptions than DD schemes. To address practical imperfections, we model detection inefficiency on the untrusted side via null outcomes and derive key rate expressions for both post-selected and non-post-selected scenarios. Our analysis reveals that secure key generation remains possible with detection efficiencies as low as 74.5\% under ideal visibility, surpassing typical DI-QKD detection efficiency thresholds~\cite{Pironio2009, Masanes2011, Woodhead2021, Branciard2012} and emphasizing the protocol’s practical advantage in lossy settings.

The manuscript is organized as follows. Section~\ref{CJWRmotivation} motivates the use of the CJWR inequality in the context of one-sided device-independent quantum key distribution along with outlining the general framework of 1sDI-QKD. In Section~\ref{securityproof}, we present the security proof under optimal collective attacks. Section~\ref{robustness} discusses the robustness of the CJWR-based 1sDI-QKD protocol, while Section~\ref{DetectionEfficiency} examines the effects of detection inefficiency. Finally, Section~\ref{Salient} highlights the salient features of our approach and outlines future research directions.

\section{CJWR-based 1sDI-QKD Protocol}\label{CJWRmotivation}

Entanglement-based QKD provides inherent security based on quantum mechanics. The BBM92 scheme~\cite{BBM92} depends on strong measurement correlations without utilizing nonlocality, whereas the Ekert91 protocol~\cite{Ekert1991} certifies key security through violations of a Bell inequality. Due to its sensitivity to loss and detection inefficiencies, DI-QKD~\cite{Acin2007, Pironio2009, Scarani2009, LiuDI-QKDExpt} undermines trust in all devices but is experimentally demanding.
The 1sDI-QKD~\cite{Branciard2012, Tomamichel1sDIQKD2013}, where only one party’s device is trusted (typically Bob’s), provides a practical alternative. Here, quantum steering~\cite{Schrödinger1935,WisemanSteering2007}, an intermediate form of nonclassicality, enables security certification against an untrusted device, making 1sDI-QKD more tolerant to experimental imperfections.

Quantum Steering~\cite{Schrödinger1935, WisemanSteering2007} is a form of quantum correlation that lies between entanglement and Bell nonlocality. It captures the ability of one party to influence the conditional state of another through local measurements nonlocally. A bipartite state is said to be steerable when it cannot be described by a local hidden state (LHS) model, where the trusted party's outcomes arise from a pre-existing quantum ensemble independent of the untrusted party’s measurement choice. Violations of steering inequalities~\cite{CavalcantiExptCriteriaSteering2009, ChenCJWR2013, Schneeloch2013, Pramanik2014, Saunders2010EPRSteering} thus serve as one-sided device-independent witnesses of quantumness. Among these, the CJWR inequality~\cite{CavalcantiExptCriteriaSteering2009} offers a symmetric and experimentally~\cite{Saunders2010EPRSteering, BennetExptSteering2012} friendly criterion for detecting steering in two-qubit systems, making it particularly suited for 1sDI-QKD protocols.


\begin{figure*}[ht!]
    \centering
    \includegraphics[width=0.75\linewidth]{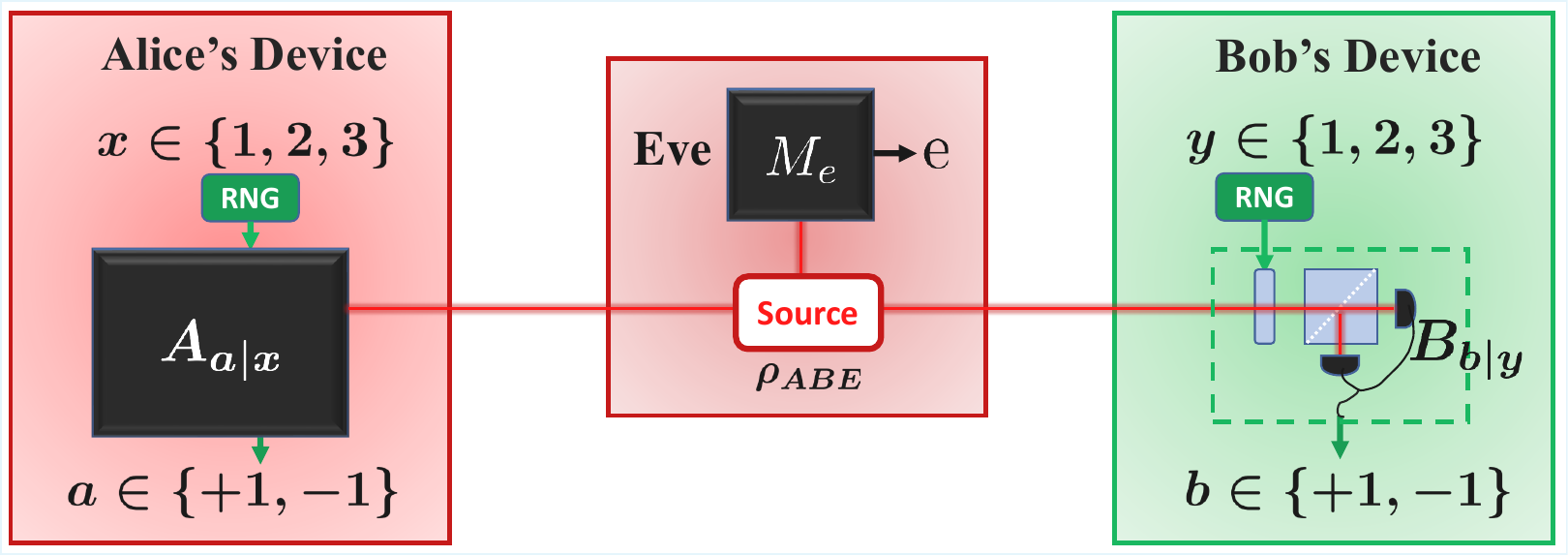}
    \caption{ Schematic of a one-sided device-independent QKD protocol. A source distributes entangled two-qubit states to Alice and Bob. Alice’s device is untrusted (black box), while Bob’s is fully trusted. Inputs \(x, y \in \{1,2,3\}\) are chosen using trusted random number generators, yielding binary outcomes \(a, b \in \{+1,-1\}\). Security is certified via quantum steering, e.g., violation of the CJWR inequality.}
    \label{1SDIfig}
\end{figure*}

For $n$ measurement settings per party, the CJWR steering function is defined as
\begin{equation}
    \mathcal{F}_n(\rho, \mu) = \frac{1}{\sqrt{n}} \left| \sum_{i=1}^{n} \langle A_i \otimes B_i \rangle \right| \leq 1,
    \label{CJWR}
\end{equation}
where $A_i = \hat{u}_i \cdot \vec{\sigma}$, $B_i = \hat{v}_i \cdot \vec{\sigma}$, and $\mu = \{\hat{u}_i; \hat{v}_i\}_{i=1}^{n}$ specifies the measurement directions, with $\hat{v}_i$ orthonormal and $\hat{u}_i$ unit vectors in $\mathbb{R}^3$.

In two-qubit systems, the two-setting CJWR inequality $\mathcal{F}_2$ fails to identify certain steerable states that remain Bell local~\cite{Saunders2010EPRSteering, CostaQuantificationSteering2016}. To detect such states and highlight the distinction between steering and Bell nonlocality, we consider the three-setting measure $\mathcal{F}_3$. This steering function can alternatively be expressed in terms of the singular values $\lambda_1, \lambda_2, \lambda_3$ of the correlation matrix $T$ associated with $\rho_{AB}$, where the matrix elements are defined as $t_{ij} = \mathrm{Tr}[(\sigma_i \otimes \sigma_j) \rho_{AB}]$. These singular values characterize the strength of quantum correlations between the two subsystems along orthogonal measurement directions. In this form, the steering function becomes
\begin{equation}
\mathcal{F}_3 = \sqrt{ \lambda_1^2 + \lambda_2^2 + \lambda_3^2 }.
\end{equation}
A violation of the bound, i.e., $\mathcal{F}_3 > 1$, confirms that the state $\rho_{AB}$ is steerable from Alice to Bob. Such violations serve as a practical witness for steering-based quantum correlations and form the foundation for establishing security in 1sDI-QKD protocols.

The three-setting CJWR inequality shares a symmetry structure similar to the CHSH inequality but captures a different class of nonclassical correlations, particularly relevant in one-sided device-independent scenarios. Its experimental applicability and robustness motivate the adoption of the CJWR function as the steering witness in our 1sDI-QKD protocol, as detailed below.

\textit{Protocol Overview:} Alice and Bob share the maximally entangled Bell state \( |\phi^+\rangle = \frac{1}{\sqrt{2}}(|00\rangle + |11\rangle) \), i.e., \( \rho_{AB} = |\phi^+\rangle \langle \phi^+| \in \mathbb{C}^2 \otimes \mathbb{C}^2 \). Bob's measurement device is trusted and fully characterized, while Alice’s is treated as a black box. Both parties randomly choose inputs \( x, y \in \{1,2,3\} \), corresponding to Pauli observables: \( A_1 = \sigma_x \), \( A_2 = -\sigma_y \), \( A_3 = \sigma_z \) for Alice; and \( B_1 = \sigma_x \), \( B_2 = \sigma_y \), \( B_3 = \sigma_z \) for Bob (See Fig.\ref{1SDIfig}).

Security is certified through the violation of the CJWR steering inequality for 3-setting ($n=3$) Eq.~\eqref{CJWR} :
\begin{equation}\label{CJWRineq}
\mathcal{F}_3 = \frac{1}{\sqrt{3}} \left| \sum_{i=1}^3 \langle A_i \otimes B_i \rangle \right| \leq 1,
\end{equation}
with $\mathcal{F}_3 > 1$ indicating steerability from Alice to Bob despite Alice’s untrusted device.

Only the rounds in which both parties measure in the \( \sigma_z \) basis (i.e., \( x = y = 3 \)) are used for raw key generation, and the corresponding outcomes are kept secret. In contrast, the outcomes from rounds where the measurement settings span all three Pauli bases (\( x, y \in \{1,2,3\} \)) are publicly disclosed and used solely for evaluating the steering parameter \( \mathcal{F}_3 \). This separation between security estimation and key extraction prevents basis mismatch and simplifies the key rate analysis. Moreover, by disclosing outcomes only in the non-key-generating rounds, the protocol limits information leakage and aligns structurally with Bell-based DI-QKD protocols, such as that of Acín \textit{et al.}~\cite{Acin2007}, enabling a direct comparison of steering- and Bell-based security frameworks. The quantum bit error rate (QBER) quantifies the probability that Alice and Bob obtain different outcomes when measuring in the same basis. Specifically, the QBER is defined as,

\begin{equation}
\label{QBER}
    Q = p(a_3 \ne b_3 \mid A_3, B_3),
\end{equation} 
representing the probability that their outcomes disagree when both measure observables \( A_3 \) and \( B_3 \), which ideally should yield identical results in the absence of noise or eavesdropping.

After parameter estimation, Alice and Bob proceed to the classical post-processing stage. They first perform \emph{error correction} over an authenticated classical channel to reconcile discrepancies in their raw keys. This is followed by \emph{privacy amplification}, typically using universal hashing~\cite{Renner2008}, to compress the reconciled key and remove any partial information available to an eavesdropper. The amount of extractable secret key is directly determined by the measured QBER and the observed steering violation \( F_3 \), ensuring composable security even in the presence of device imperfections on Alice's side.

\section{Security Analysis Against Collective Attacks}\label{securityproof}

In our protocol, only Bob’s measurement device is assumed to be trusted and well characterized, while Alice’s device is treated as completely untrusted. This setting defines the 1sDI scenario, where the security of the key must be established without relying on any knowledge about Alice’s internal operations. Instead, all security claims are based entirely on the observed correlations between Alice and Bob’s measurement outcomes.

Although this model is less restrictive than DI-QKD, it still provides strong security guarantees. In particular, it avoids certain experimental challenges such as the detection loophole, which limits the practicality of device-independent approaches~\cite{Pironio2009}. In what follows, we show how the violation of the CJWR steering inequality, along with the measured QBER, can be used to certify the presence of secure correlations and to derive a bound on the secret key rate. The CJWR inequality violation serves as evidence of nonclassical steering correlations, assuming trust only in Bob's measurement device. Together, these quantities allow Alice and Bob to estimate Eve's accessible information.

We assume that Alice, Bob, and Eve share a tripartite pure state
\(\ket{\Phi_{\text{ABE}}} \in \mathcal{H}_{\text{A}}^{\otimes N} \otimes \mathcal{H}_{\text{B}}^{\otimes N} \otimes \mathcal{H}_{\text{E}}
\), where \(N\) is the number of rounds used for key generation. Without loss of generality, we take the local Hilbert spaces to be of equal finite dimension, i.e., \(\mathcal{H}_{\text{A}} \simeq \mathcal{H}_{\text{B}} \simeq \mathbb{C}^{d}\).

For the security analysis, we assume that Eve performs a collective, or i.i.d., attack~\cite{Devetak2005, Acin2006njp, Acin2006, Acin2007, Pironio2009, Barrett2005a, KonigMinMaxEntropy2009, Masanes2011, Vazirani14, Tomamichel2012FiniteKey, FriedmanEAT2018, Dupuis2020EAT, Scarani2009}. In this setting, the state and measurement procedure used by Eve are the same in every round and independent across rounds. As a result, the total shared state between Alice, Bob, and Eve takes the form \(
\ket{\Phi_{\text{ABE}}} = \ket{\phi_{\text{ABE}}}^{\otimes N},\)
where \(\ket{\phi_{\text{ABE}}}\) is the state shared in a single round of the protocol. In addition, we assume that the devices are memoryless, meaning the measurement in each round depends only on the current input and not on previous rounds.

We adopt a one-way classical post-processing strategy from Bob to Alice, following the approach of Refs.~\cite{Acin2006njp, Acin2006, Acin2007,Pironio2009}. Under this setting, the asymptotic key rate can be bounded using the Devetak–Winter formula\cite{Devetak2005}:
\begin{equation}\label{1sDIKeyrate}
    r^{\text{1sDI}} \geq I(A_3:B_3) - \chi(B_3:E),
\end{equation}
where \(I(A_3:B_3)\) is the mutual information between Alice and Bob, and \(\chi(B_3:E)\) is the Holevo quantity, representing an upper bound on the information accessible to Eve about Bob’s outcomes.

The choice of Bob-to-Alice communication is particularly advantageous in our protocol, as also discussed in Ref.~\cite{Acin2006njp}. Since Bob's device is trusted, this direction of classical post-processing allows for a tighter bound on Eve's accessible information. Consequently, the relevant leakage term in the key rate expression is \(\chi(B_3:E)\), rather than \(\chi(A_3:E)\), which would apply if Alice's data were revealed in the classical post-processing step.

Our main objective is to derive an upper bound on Eve's accessible information, quantified by the Holevo quantity \(\chi(B_3:E)\). To achieve this, we follow a sequence of steps outlined below:

\textit{Step 1:} To compute a tight upper bound on Eve's accessible information, we simplify the analysis using techniques similar in spirit to those introduced in Ref.~\cite{Pironio2009}. In the one-sided device-independent scenario, Bob's measurement device is fully trusted and assumed to perform Pauli measurements. Therefore, even if Eve prepares a general state in \(\mathbb{C}^d \otimes \mathbb{C}^d\), Bob effectively accesses only a two-dimensional Hilbert space. As a result, without loss of generality, we can assume that the shared state distributed by Eve is a qudit-qubit state. Certain symmetries inherent in the CJWR inequality can be used to simplify the analysis by reducing the effective dimension of the shared state. This is formalized in the following lemma.

\begin{lemma}[Reduction to a two-qubit subspace]
Let \( \rho \in \mathbb{C}^d \otimes \mathbb{C}^2 \) be a bipartite quantum state, where Bob performs projective measurements along the Pauli directions \( \sigma_1, \sigma_2, \sigma_3 \), and Alice performs Hermitian dichotomic observables \( A_1, A_2, A_3 \) satisfying \( A_l^2 = \mathbb{I} \). Then, the quantum steering expression
\begin{equation}
\mathcal{F}_3(\rho) := \frac{1}{\sqrt{3}} \left| \sum_{l=1}^3 \langle A_l \otimes \sigma_l \rangle_\rho \right|
\end{equation}
is bounded by \( \mathcal{F}_3(\rho) \leq \sqrt{3} \). The bound is tight if and only if the observables \( A_l \) mutually anticommute. In such cases, the optimal value is achieved in a two-qubit system.
\end{lemma}

\begin{proof}
Define the CJWR operator:
\begin{align*}
\mathcal{B}_{\mathrm{CJWR}} := \sum_{l=1}^3 A_l \otimes \sigma_l.    
\end{align*}
Expanding the square of the CJWR operator,
\begin{align*}
\mathcal{B}_{\mathrm{CJWR}}^2 &= \left( \sum_{l} A_l \otimes \sigma_l \right)^2 \nonumber \\
&= 3\, \mathbb{I} \otimes \mathbb{I} + \sum_{l < m} [A_l, A_m] \otimes \sigma_l \sigma_m,
\end{align*}
where we used \( A_l^2 = \mathbb{I},  \sigma_l^2 = \mathbb{I} \), and the anticommutation relations \( \{ \sigma_l, \sigma_m \} = 0 \) for \( l \neq m \), which implies \( \sigma_m \sigma_l = -\sigma_l \sigma_m \).

If the observables mutually anticommute, i.e., \( \{A_l, A_m\} = 0 \) for \( l \neq m \), then \( [A_l, A_m] = 2 A_l A_m \). Hence,
\begin{align*}
    \mathcal{B}_{\mathrm{CJWR}}^2 = 3\, \mathbb{I} \otimes \mathbb{I} + 2 \sum_{l < m} A_l A_m \otimes \sigma_l \sigma_m.
\end{align*}
each term \( A_l A_m \otimes \sigma_l \sigma_m \) has operator norm at most 1, since all factors are unitary. 

We use the operator norm (or spectral norm) to bound the CJWR steering operator. For a Hermitian operator \( O \), the operator norm is defined as
\begin{align*}
    \| O \|_\infty = \sup_{\| \psi \| = 1} | \langle \psi | O | \psi \rangle |,
\end{align*}
which corresponds to the largest eigenvalue in magnitude of \( O \). Using this, we obtain the following bound on the CJWR operator:
\begin{align*}
\| \mathcal{B}_{\mathrm{CJWR}}^2 \|_\infty \leq 3 + 6 = 9, \quad \Rightarrow \quad \| \mathcal{B}_{\mathrm{CJWR}} \|_\infty \leq 3.
\end{align*}

An explicit example achieving the maximum is given by choosing \( A_l = \sigma_l \) and taking \( \rho \) as the maximally entangled singlet state, for which \( \mathcal{F}_3(\rho) = \sqrt{3} \)~\cite{CavalcantiExptCriteriaSteering2009, CostaQuantificationSteering2016}.

Our proof strategy mirrors the dimensionality reduction argument employed in the CHSH scenario~\cite{Masanes2006}, where any two dichotomic observables with eigenvalues \( \pm 1 \) are shown to generate a two-dimensional invariant subspace. In our case, the set of three mutually anticommuting observables \( A_1, A_2, A_3 \) similarly generate a representation of the real Clifford algebra \( \mathrm{Cl}_3(\mathbb{R}) \), whose minimal irreducible representation is two-dimensional~\cite{Lounesto2001}. This justifies the restriction to a two-qubit system without loss of generality. Hence, there exists a subspace \( \mathcal{H}_2 \subseteq \mathcal{H}_A \) such that \( A_l \) act as \( \sigma_l \) on \( \mathcal{H}_2 \cong \mathbb{C}^2 \).

Define the projected state \(\rho_{\mathrm{eff}} := (P \otimes \mathbb{I}) \rho (P \otimes \mathbb{I}) \in \mathbb{C}^2 \otimes \mathbb{C}^2,\)
where \( P: \mathcal{H}_A \to \mathcal{H}_2 \) is the projection. Since the CJWR operator acts trivially outside this subspace,
\begin{align*}
\mathcal{F}_3(\rho) = \mathcal{F}_3(\rho_{\mathrm{eff}}).    
\end{align*}

This completes the proof that the optimal value of the CJWR expression is achieved within a two-qubit subspace, and any higher-dimensional scenario does not offer an advantage.

\renewcommand{\qedsymbol}{}
\end{proof}

\textit{Step 2:} In the previous step, we argued that, without loss of generality, Eve can restrict herself to preparing a bipartite state in \( \mathbb{C}^2 \otimes \mathbb{C}^2 \), where Alice’s measurements are fixed to be qubit Pauli observables. We now investigate which class of \( \mathbb{C}^2 \otimes \mathbb{C}^2 \) states enables Eve to extract the maximum possible information while maintaining an optimal value of the CJWR expression between Alice and Bob. In the following lemma, we show that any such two-qubit state \( \rho \in \mathbb{C}^2 \otimes \mathbb{C}^2 \) can be reduced to a Bell-diagonal form without affecting the CJWR value.
  
\begin{lemma}[Reduction to Bell-diagonal form]
\label{lemma:bell_diagonal}
Let \(\rho \in \mathbb{C}^2 \otimes \mathbb{C}^2\) be a two-qubit state shared between Alice and Bob, and suppose Bob performs fixed Pauli measurements \( \sigma_x, \sigma_y, \sigma_z \). Then there exists a Bell-diagonal state \( \rho_{\Lambda} \) such that the CJWR steering expression
\[
\mathcal{F}_3(\rho) = \frac{1}{\sqrt{3}} \left| \sum_{i=1}^{3} \langle A_i \otimes \sigma_i \rangle_\rho \right|
\]
remains unchanged:
\[
\mathcal{F}_3(\rho) = \mathcal{F}_3(\rho_{\Lambda}).
\]
moreover, \( \rho_{\Lambda} \) can be obtained from \( \rho \) by applying a symmetrization under conjugation by \( \sigma_y \otimes \sigma_y \) followed by taking the real part.
\end{lemma}
\begin{proof}
The CJWR functional depends only on the two-point correlators \( \langle A_i \otimes \sigma_i \rangle \), and not on local marginals or off-diagonal coherences in other Bell-state bases. Consider the symmetrized state:
\[
\bar{\rho} = \frac{1}{2} \left[ \rho + (\sigma_y \otimes \sigma_y)\rho(\sigma_y \otimes \sigma_y) \right].
\]
This operation preserves all correlators of the form \( \langle A_i \otimes \sigma_i \rangle \), since the Pauli matrices \(\sigma_i\) are either invariant or change sign under conjugation by \(\sigma_y\), and \(A_i\) can be redefined accordingly. As a result,
\[
\mathcal{F}_3(\rho) = \mathcal{F}_3(\bar{\rho}).
\]
Moreover, this conjugation eliminates off-diagonal elements connecting Bell states with opposite \(\sigma_y \otimes \sigma_y\) eigenvalues.

To remove the remaining imaginary parts of the off-diagonal terms, we take the real part:
\begin{align}
\rho_\Lambda & = \frac{1}{2} (\bar{\rho} + \bar{\rho}^*)\nonumber \\
& = \begin{pmatrix}
\Lambda_{\Phi^+} & & & \\
& \Lambda_{\Psi^-} & & \\
& & \Lambda_{\Phi^-} & \\
& & & \Lambda_{\Psi^+}
\end{pmatrix} 
\label{BellDiagonal}
\end{align}
yielding a real, symmetric state diagonal in the Bell basis (i.e. \(\{\ket{\Phi^+}, \ket{\Psi^-}, \ket{\Phi^-}, \ket{\Psi^-} \}\)). Since both steps preserve the relevant correlators, we have:
\[
\mathcal{F}_3(\rho) = \mathcal{F}_3(\rho_{\Lambda}),
\]
and thus it suffices to restrict the security analysis to Bell-diagonal states.

Following the symmetrization, the state $\rho$ can be locally rotated within the $(x,z)$ plane to arrange the Bell-state eigenvalues in a fixed order \cite{Pironio2009}, such as $\Lambda_{\Phi^+} \geq \Lambda_{\Psi^-}$ and $\Lambda_{\Phi^-} \geq \Lambda_{\Psi^+}$, without affecting the CJWR functional. These rotations are unitary operations that preserve two-qubit correlators $\langle A_i \otimes \sigma_i \rangle$, and thus leave $\mathcal{F}_3(\rho)$ invariant.

\renewcommand{\qedsymbol}{}
\end{proof}

\textit{Step 3:} Without loss of generality, Eve can send any mixture of Bell-diagonal states like \(\rho_{AB} = \sum_{\Lambda} p_{\Lambda}\rho_{\Lambda}\), where \(\Lambda\) is a classical ancilla known to her. Now, we need to calculate the Holevo bound \(\chi_{\Lambda}(B_3 : E)\) for that Bell-diagonal state. 
For the Bell diagonal state the Holevo bound \(\chi_{\Lambda}(B_3:E)\) can be calculated as, 
\begin{equation}
\begin{aligned}
    \chi_{\Lambda}(B_3 : E) &= S(\rho_E) - \sum_{b_3 = \pm 1} p(b_3)\, S(\rho_{E|b_3}) \\
    &= -\sum_{i=1}^4 \Lambda_i \log_2 \Lambda_i - \frac{1}{2} \left( S(\rho_{E|+1}) + S(\rho_{E|-1}) \right),
\end{aligned}
\end{equation}
here \(\Lambda_1 = \Lambda_{\Phi^+}, \Lambda_2 = \Lambda_{\Psi^-}, \Lambda_3 = \Lambda_{\Phi^-}, \Lambda_4 = \Lambda_{\Psi^+}\).

To determine a secure key rate in a 1sDI-QKD, we use Bob's \(\sigma_z\) measurement to assess Eve's accessible information, \(\chi_{\Lambda}(B_3: E)\). This decision is motivated by the necessity for device-independent security proofs to provide a strong, worst-case bound on Eve's knowledge that is valid for all attacks compatible with the observed steering or CJWR violation. The upper bound on \(\chi_{\Lambda}(B_3 : E)\) is given by,
\begin{equation}
\label{ChiUpper}
    \chi_{\Lambda}(B_3 : E) \leq -\sum_{i=1}^4 \Lambda_i \log_2 \Lambda_i - h(\Lambda_{1} + \Lambda_3).
\end{equation}
Here \(h(x) = -x \log_2 x - (1 - x) \log_2 (1 - x),\) \;\;h(x) is the binary entropy. The upper bound in Eq.~(\ref{ChiUpper}) is tight for Bell-diagonal states when Bob measures \(\sigma_z\), which is the case where Eve gains maximal information (see Lemma 5 of Ref.~\cite{Pironio2009}). Although Bob’s trusted apparatus allows flexibility in choosing an optimal measurement basis for key generation (e.g., \(\vec{\sigma}\cdot\hat{n}\)), analyzing \(\sigma_z\) ensures a conservative, analytically tractable security proof that avoids reliance on numerical optimization and guarantees robustness across implementations. This approach, standard in DI and 1sDI-QKD, strengthens the security analysis by addressing the worst-case scenario, ensuring the key remains secure as long as the observed Bell violation exceeds the classical threshold.

We use the entropic inequality given in Lemma 6 of Ref.~\cite{Pironio2009} to upper-bound Eve's accessible information. It states that for a Bell-diagonal state with eigenvalues \(\Lambda_{1}, \Lambda_{2}, \Lambda_{3}, \Lambda_{4}\) (all \(\geq 0\) and summing to 1), \( \mathcal{R}^2 = (\Lambda_{1} - \Lambda_{2})^2 + (\Lambda_{3} - \Lambda_{4})^2,\) and taking into account \(S(\Lambda) =-\sum_{i=1}^4 \Lambda_i \log_2 \Lambda_i - h(\Lambda_{1} + \Lambda_{3})\),  the following inequality holds:

\begin{align}
    \mathcal{S}(\Lambda) 
\leq h\left( \frac{1 + \sqrt{2\mathcal{R}^2 - 1}}{2} \right)
\quad \text{if } \mathcal{R}^2 > \frac{1}{2},
\end{align}
\begin{align}
\mathcal{S}(\Lambda) \leq 1 \quad \text{if } \mathcal{R}^2 \leq \frac{1}{2},
\end{align}
with equality in the first bound if and only if either \( \Lambda_{1,3} = 0 \) or \( \Lambda_{2,4} = 0 \)~\cite{Acin2007}.

\textit{Step 4:} We now relate Eve's accessible information to the CJWR function \( \mathcal{F}_3(\rho) \). For a Bell-diagonal state \( \rho_{\Lambda} \), the correlation matrix \( T^{\Lambda} \) is diagonal in the Pauli basis, with entries given by
\(
T^{\Lambda}_{11} = \Lambda_1 - \Lambda_2 - \Lambda_3 + \Lambda_4,\quad
T^{\Lambda}_{22} = -\Lambda_1 - \Lambda_2 + \Lambda_3 + \Lambda_4,\quad
T^{\Lambda}_{33} = \Lambda_1 - \Lambda_2 + \Lambda_3 - \Lambda_4.
\)
The optimal CJWR value for such a state is
\begin{align}
    \mathcal{F}_3^{\Lambda} = \sqrt{(T^{\Lambda}_{11})^2 + (T^{\Lambda}_{22})^2 + (T^{\Lambda}_{33})^2}.
    \label{F_3L}
\end{align}

For the optimal collective attack, the entropic bound of Ref.~\cite{Acin2007}, as mentioned in step~3,
is saturated by extremal Bell-diagonal states satisfying either
$\Lambda_1=\Lambda_3=0$ or $\Lambda_2=\Lambda_4=0$.
Without loss of generality, we consider the first case and write
\begin{align*}
  \Lambda_1 &= \Lambda_3 = 0, \qquad
  \Lambda_2 = k, \qquad
  \Lambda_4 = 1-k ,
\end{align*}
with $k\in[0,1]$.
Then the parameter $R$ entering the entropic bound is then
\begin{align*}
  R^2 &= (\Lambda_1-\Lambda_2)^2 + (\Lambda_3-\Lambda_4)^2   \notag\\
      &= k^2 + (1-k)^2 .
\end{align*}
Using the correlation matrix elements of a Bell-diagonal state, one finds
\begin{align*}
  T_{11}^{\Lambda} &= 1-2k, \qquad
  T_{22}^{\Lambda} = 1-2k, \qquad
  T_{33}^{\Lambda} = -1 ,
\end{align*}
so that the CJWR steering value of Eq.~(\ref{F_3L}) becomes
\begin{align*}
  \mathcal (F_3^{\Lambda})^2
  &= 2(1-2k)^2 + 1 \notag \\
   &= 8k^2 - 8k + 3 \notag \\
   &= 4(k^2+(1-k)^2) -1.
\end{align*}
On the other hand,
\begin{align*}
  4R^2 - 1
  &= 4\big[k^2+(1-k)^2\big] - 1      \notag\\
  &= 8k^2 - 8k + 3 ,
\end{align*}
which yields
\begin{equation*}
  \mathcal (F_3^{\Lambda})^2 = 4R^2 - 1 .
\end{equation*}
The complementary extremal case $\Lambda_2=\Lambda_4=0$ can be treated
analogously and leads to the same relation.
Hence,
\begin{equation*}
  R = \frac{1}{2}\sqrt{1+\mathcal (F_3^{\Lambda})^2},
\end{equation*}
demonstrating that the intermediate parameter $R$ is fully determined by the observed CJWR steering violation $\mathcal F_3^{\Lambda}$.

This can be compactly written as \( \mathcal{F}_3^{\Lambda} = \sqrt{4\mathcal{R}^2 - 1} \), where \( \mathcal{R} \) is defined in terms of the Bell-state probabilities and the threshold \( \mathcal{R} = 1/\sqrt{2} \) corresponds to the classical limit for the CJWR inequality, analogous to the CHSH case in Ref.~\cite{Pironio2009}.

Since Eve's most general collective strategy can involve preparing a mixture of Bell-diagonal states, we consider \( \rho_{AB} = \sum_{\Lambda} p_{\Lambda} \rho_{\Lambda} \). In this case, her accessible information is given by
\begin{equation}
    \chi(B_3 : E) = \sum_{\Lambda} p_{\Lambda} \chi_{\Lambda}(B_3 : E).
\end{equation}
using the entropic bound from Step 3, \( \chi_{\Lambda}(B_3 : E) \leq \mathcal{S}(\mathcal{F}_3^{\Lambda}) \), and the fact that \( \mathcal{S}(\cdot) \) is concave, we obtain
\begin{equation*}
    \chi(B_3 : E) \leq \sum_{\Lambda} p_{\Lambda} \mathcal{S}(\mathcal{F}_3^{\Lambda}) \leq \mathcal{S}\left( \sum_{\Lambda} p_{\Lambda} \mathcal{F}_3^{\Lambda} \right).
\end{equation*}
Moreover, by convexity of the CJWR expression (see Eq.~\eqref{CJWRineq}) and the triangle inequality (\( \left| a + b \right| \leq \left| a \right| + \left| b \right|
\)), it holds that
\(
\mathcal{F}_3(\rho_{AB}) \leq \sum_{\Lambda} p_{\Lambda} \mathcal{F}_3(\rho_{\Lambda}),
\)
which implies
\begin{equation*}
    \chi(B_3 : E) \leq \mathcal{S}\left( \mathcal{F}_3(\rho_{AB}) \right).
\end{equation*}

Assuming uniform marginals for Alice and Bob, the mutual information becomes \( I(A_3 : B_3) = 1 - h(Q) \), where \( Q \) is QBER. Combining all steps, the key rate under optimal collective attacks in the 1sDI scenario (see Eq.~\eqref{1sDIKeyrate}) is lower bounded by
\begin{equation}
    \label{Opt1sDIKeyrate}
    r^{1s\mathrm{DI}} \geq 1 - h(Q) - h\left( \frac{1 + \sqrt{(\mathcal{F}_3^2 - 1)/2}}{2} \right).
\end{equation}

\section{Noise tolerance of CJWR-based 1sDI-QKD}\label{robustness}
To quantitatively compare the noise tolerance of 1sDI-QKD against DI and DD protocols, we consider a widely used noise model in the QKD literature~\cite{Scarani2009,
PirandolaQKDReview2020, PortmannQKDReview2022, RennerQKDReview2023, Acin2007,Pironio2009,Branciard2012}. In this model, the maximally entangled Bell state \(\ket{\phi^+}\) is subjected to depolarizing noise, resulting in the mixed state
\begin{equation}
\rho_\nu = \nu \ket{\phi^+}\bra{\phi^+} + (1 - \nu)\frac{\mathbb{I}}{4},
\label{rho_werner}
\end{equation}
where \(\nu \in [0,1]\) denotes the visibility, quantifying the strength of the noise. To evaluate the secret key rate achievable in the 1sDI-QKD scenario, we employ Eq.~(\ref{Opt1sDIKeyrate}) and compute the relevant quantities from this state. The QBER is given by
\(Q = \frac{1 - \nu}{2}\), as obtained from Eq.~(\ref{QBER}), while the CJWR correlator evaluates to
\(\mathcal{F}_3 = \nu \sqrt{3}.\) We can rewrite this \((Q, \mathcal{F}_3)\) as a correlation,
\begin{align}
    \mathcal{F}_3 = \sqrt{3} (1 - 2 Q).
    \label{WernerCorr}
\end{align}
This relation in Eq.~(\ref{WernerCorr}) is independent of any assumptions on the source or Alice’s measurement device, relying solely on the observed statistics \(Q, \mathcal{F}_3\).  

For the DI and DD scenarios, we use the corresponding key rate expressions derived in previous works. In the DI case, the key rate is bounded using the observed violation of a Bell inequality, typically the CHSH inequality, following the approach of Ref.~\cite{Acin2007}. The relevant key rate expression is a function of the CHSH parameter \(\mathcal{B}\) and QBER \(Q\), which, under depolarizing noise, takes the form \(\mathcal{B} = 2\sqrt{2} \nu\) and \(Q = \frac{1-\nu}{2}\). In the CHSH scenario,    the correlation is given by \(\mathcal{B} = 2\sqrt{2}(1 - 2 Q)\)\cite{Acin2007}. 
In contrast, the key rate for the device-dependent (DD) scenario~\cite{Acin2007} is computed under the assumption of full trust in both the state preparation and measurement devices. In this setting, the Devetak--Winter formula~\cite{Devetak2005} applies directly, with the secret key rate determined by the mutual information between Alice and Bob and the conditional entropy of Eve.

\begin{figure}[h!]
    \centering
    \includegraphics[width=\linewidth]{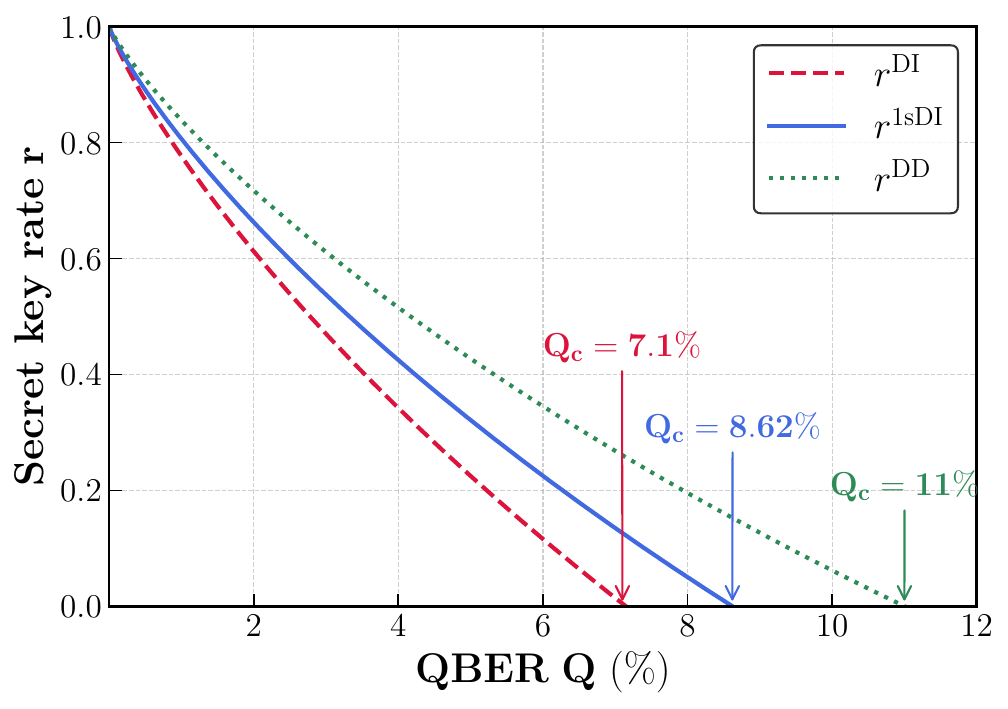}
    \caption{\footnotesize Comparison of key rates (r) as a function of the QBER (Q). The red dashed line represents the key rate \( r^{\mathrm{DI}} \) in the DI scenario based on Bell inequality violation. The blue solid line corresponds to the 1sDI key rate \( r^{\mathrm{1sDI}} \) certified via CJWR steering inequality violation. The green dotted line shows the DD key rate \( r^{\mathrm{DD}} \), where both parties' devices are trusted. }
    \label{Key rate robustness}
\end{figure}

To enable a consistent and transparent comparison across different security models, we evaluate the key rates for the DD, 1sDI, and DI scenarios using a common depolarizing noise model \(\rho_\nu\), parameterized by the visibility \(\nu\). The resulting key rates, plotted as functions of the QBER in Fig.~\ref{Key rate robustness}, reveal distinct noise thresholds for each protocol. For the fully DI-QKD protocol, the critical QBER is \(Q_c^{\mathrm{DI}} = 7.1\%\)~\cite{Acin2007, Pironio2009, Scarani2009}, while in the DD scenario it increases to approximately \(11\%\)~\cite{Shor2000, Acin2007, Pironio2009}. Our CJWR-based 1sDI-QKD protocol achieves a critical QBER of \(Q_c^{\mathrm{1sDI}} = 8.62\%\), which lies between these two regimes:
\begin{equation}
Q_c^{\mathrm{DI}} < Q_c^{\mathrm{1sDI}} < Q_c^{\mathrm{DD}}.
\end{equation}

This intermediate robustness highlights the advantage of the 1sDI setting, which tolerates more noise than fully device-independent protocols while considering fewer assumptions than fully device-dependent approaches. Furthermore, our protocol compares favorably with other steering- or nonlocality-based schemes: for instance, a DI-QKD protocol based on three-setting Bell inequalities yields a threshold of \(Q_c = 7.5\%\)~\cite{Masanes2011}, while one employing an asymmetric Bell inequality reports \(Q_c = 8.34\%\)~\cite{Woodhead2021}.

\begin{table*}[t]
\centering
\small
\renewcommand{\arraystretch}{1.2}
\begin{tabular}{p{2.7cm} p{2.6cm} p{2.4cm} p{4.3cm} p{4.0cm}}
\hline\hline
\textbf{Reference}
& \textbf{Protocol type}
& \textbf{Attack model}
& \textbf{Role of steering}
& \textbf{Key features} \\
\hline

Branciard et al.~\cite{Branciard2012}
& Prepare-measure (BBM92-type), 1sDI
& Coherent (memoryless)
& Implicit: steering required for security when the protocol is expressed in an entanglement-based picture
& Eve’s information bounded via QBERs in two bases using entropic uncertainty relations; reduces to standard DD key rate under symmetric noise \\

Tomamichel et al.~\cite{Tomamichel1sDIQKD2013}
& Prepare-measure (BB84-type), 1sDI
& Coherent, finite-size
& Implicit: steering arises through monogamy-of-entanglement games
& Fully composable finite-size security; very low key rates, positive only at very low QBER \\

Pramanik et al.~\cite{Pramanik2014}
& Entanglement-based, 1sDI
& Individual attacks
& Explicit: violation of a steering inequality guarantees security
& Provides a lower bound on the key rate from steering violation; no explicit Holevo-information bound \\

Mukherjee et al.~\cite{Mukherjee2023SteeringQKD}
& Entanglement-based, 1sDI
& Asymptotic
& Explicit: CJWR steering violation combined with local filtering
& Security is characterized via an effective QBER after filtering; Eve’s information not bounded directly \\

Masini \& Sarkar~\cite{masini20241sDIQKD}
& Prepare--measure (BB84-type), 1sDI
& Coherent, finite-size
& Implicit: steering used to bound Eve’s entropy via SDP and EAT
& Fully composable; numerical bounds only, no closed-form analytic dependence on QBER or steering \\

\textbf{This work}
& Entanglement-based (E91-type), 1sDI
& Collective
& Explicit: security certified directly by CJWR steering violation
& The analytic Holevo bound expressed in terms of the observed pair $(Q,\mathcal F_3)$; model-independent detection-efficiency thresholds \\

\hline\hline
\end{tabular}
\caption{Comparison of representative 1sDI-QKD protocols, highlighting protocol type, attack model, and the operational role of steering.}
\label{Comparison_Table}
\end{table*}

Our protocol differs conceptually from most existing 1sDI-QKD schemes summarized in Table~\ref{Comparison_Table}. Whereas earlier approaches ultimately bound Eve’s information using one or more QBER parameters, our analysis ties security directly to the observed strength of steering correlations rather than relying on error rates alone. In particular, in the protocol of Branciard \emph{et al.}~\cite{Branciard2012}, Eve’s information is bounded via an entropic uncertainty relation using QBERs measured in two complementary bases; under ideal detection efficiency and symmetric depolarizing noise, these QBERs coincide, causing the 1sDI key-rate bound to reduce to the standard device-dependent BB84/BBM92 expression~\cite{Shor2000, Pironio2009, Scarani2009}
\(r = 1 - 2h(Q),\) with a critical QBER $Q_c \simeq 11\%$. A similar QBER-centric behavior appears in the finite-size, monogamy-game-based analysis of Tomamichel \emph{et al.}~\cite{Tomamichel1sDIQKD2013}, where a positive key can be obtained only for very low depolarizing noise, $Q \lesssim 1.5\%$. Steering-based QKD protocols that do not explicitly employ a Holevo-information bound, notably those of Pramanik \emph{et al.}~\cite{Pramanik2014} (restricted to individual attacks) and Mukherjee \emph{et al.}~\cite{Mukherjee2023SteeringQKD} (with no adversarial model specified), certify security through the violation of a steering inequality. In these works, the condition for a positive key rate can be re-expressed in terms of an effective QBER, leading to a comparatively high tolerable noise threshold, as \(Q_c \simeq 21\%\). Nevertheless, since Eve’s information is not explicitly bounded, this threshold does not correspond to a quantitative characterization of Eve’s accessible information. While Masini \emph{et al.}~\cite{masini20241sDIQKD} mainly focused on demonstrating near-optimal detection-efficiency thresholds by means of SDP-based entropy bounds and a three-outcome measurement model, their approach yields purely numerical key-rate bounds without a closed-form dependence on QBER or steering violation. In contrast, our E91-style 1sDI protocol bounds Eve’s Holevo information directly from the observed pair $(Q,\mathcal{F}_3)$, leading under the same depolarizing-noise model to a nontrivial critical QBER $Q_c \simeq 8.62\%$ derived solely from measured correlations and without assuming a noise model and measurement model.

\section{Detection Efficiency in 1sDI-QKD}\label{DetectionEfficiency}

While our previous analysis assumes ideal detection conditions, realistic implementations of QKD must account for detection inefficiencies, particularly due to the well-known detection loophole~\cite{Pironio2009, Scarani2009}, which poses a major challenge for DI-QKD protocols. In such protocols, where both parties are untrusted, achieving secure key distribution requires very high detection efficiencies. Specifically, efficiencies on the order of $92.3\%$~\cite{Pironio2009} or even $94.5\%$~\cite{Masanes2011, Branciard2012} are necessary under assumptions of ideal visibility ($\nu = 1$).

The primary reason for such stringent requirements is that Eve may exploit undetected events in either party’s device to simulate nonlocal correlations. However, in the 1sDI scenario, only one party (Alice) is untrusted, while the other (Bob) uses a trusted, fully characterized measurement device. This relaxation allows for more practical implementations with comparatively lower detection efficiency thresholds.

In the 1sDI scenario, detection inefficiency manifests through no-click events on Alice's side, which we denote by the null outcome $\varnothing$. This effectively increases Alice’s output alphabet to three possible outcomes: $\{+1, -1, \varnothing\}$. Following the approach of Ref.~\cite{Pironio2009}, we address this by deterministically mapping the null outcome to $-1$, thereby reducing the measurement to a binary-output POVM. The resulting effective measurement operators on Alice's side take the form
\begin{equation}
    \left\{ \eta_A A_{+1|i},\ \eta_A A_{-1|i} + (1 - \eta_A) \mathbb{I} \right\},
\end{equation}
where $\eta_A \in [0,1]$ denotes Alice’s detection efficiency and $A_{\pm 1|i}$ are the ideal POVM elements for input $i$.

Since Bob’s device is trusted, we do not explicitly model his inefficiency and consider only those rounds in which his detector clicks. As Eve cannot exploit losses on the trusted side, this selective treatment remains secure and operationally relevant.

There are two natural ways to incorporate the detection inefficiency into the key rate analysis. In the first, we adopt a non-post-selected strategy, retaining all rounds, including those where Alice registers a null outcome. In this case, the QBER becomes
\(Q_{\text{PS}'} = \frac{1 - \nu \eta_A}{2}\), explicitly dependent on the product of the state visibility $\nu$ and Alice's detection efficiency $\eta_A$.

In contrast, a postselection-based strategy, similar to that used in Ref.~\cite{Branciard2012}, discards all rounds in which Alice does not report a valid outcome. In this case, the QBER is independent of $\eta_A$ and takes the form
\(Q_{\text{PS}} = \frac{1 - \nu}{2}\).

It is important to emphasize that although postselection may improve the observed QBER, it cannot be used when estimating Eve’s information. Post-selection can introduce side information to Eve if she has any control or knowledge over the detection process, especially since Alice’s device is untrusted. Consequently, the bound on Eve’s Holevo information must be derived from the entire ensemble of rounds, without post-selection, to preserve composability.
\begin{figure}[h!]
    \centering
    \includegraphics[width=\linewidth]{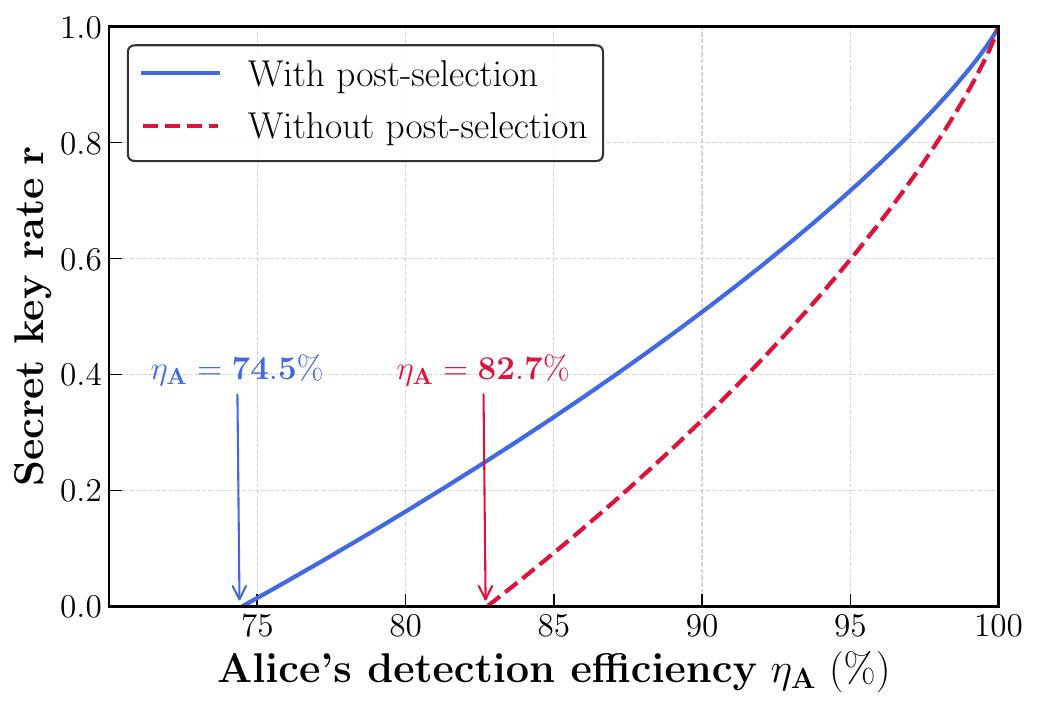}
    \caption{\footnotesize Comparison of secret key rates \( r \) as a function of Alice’s detection efficiency \( \eta_A \) under ideal visibility (\( \nu = 1 \)) for a 1sDI-QKD protocol. The red dashed curve corresponds to the key rate without post-selecting QBER (Eq.~\ref{WithoutPostselected}), while the blue solid curve represents the post-selected case (Eq.~\ref{WithPostselected}), where QBER is constant. Post-selection allows secure key generation at lower detection efficiencies, down to \( 74.5\% \), highlighting a practical advantage over fully device-independent QKD.
}
    \label{Key rate Detection eff}
\end{figure}

To model the effect of loss on the steering parameter, we adopt the null-outcome mapping described above. Under this model, the observed CJWR steering parameter becomes \(  \mathcal{F}_3 = \eta_A \nu \sqrt{3},\) indicating that detection inefficiency scales linearly with the visibility and degrades the strength of the observed steering correlations.

Using this modified expression, we can now write the corresponding key rate expressions for the two cases. When no post-selection is applied, the secure key rate becomes
\begin{equation}
    \label{WithoutPostselected}
    r^{1sDI}_{\text{PS}'} = 1 - h(Q_{\text{PS}'}) - h\left( \frac{1 + \sqrt{(\mathcal{F}_3^2 - 1)/2}}{2} \right),
\end{equation}
whereas under post-selection, the key rate is given by
\begin{equation}
    \label{WithPostselected}
    r^{1sDI}_{\text{PS}} = \eta_A (1 - h(Q_{\text{PS}})) - h\left( \frac{1 + \sqrt{(\mathcal{F}_3^2 - 1)/2}}{2} \right).
\end{equation}

These expressions provide a complete and realistic framework for evaluating the performance of our 1sDI-QKD protocol under lossy conditions. The post-selection strategy benefits from improved QBER but at the cost of reduced key throughput, while the non-post-selected version ensures data integrity at the expense of tighter efficiency requirements.

In Fig.~\ref{Key rate Detection eff}, we illustrate how the secret key rate varies with Alice’s detection efficiency $\eta_A$, assuming perfect visibility ($\nu = 1$). The red dashed curve corresponds to the case without post-selection, where a key can be generated only if $\eta_A$ exceeds $82.7\%$. The blue solid curve shows the postselected strategy, which lowers the threshold to $74.5\%$. These values are already significantly below the critical efficiencies typically required for DI-QKD, where values above $92\%$ are common~\cite{Pironio2009, Branciard2012}.

To explore how visibility impacts the security threshold, we further examine the relationship between $\nu$ and the minimum detection efficiency needed for key generation. As shown in Fig.~\ref{eta_vs_nu}, the threshold $\eta_A$ decreases with increasing visibility ($\nu$). For all values of $\nu$, the post-selected strategy performs better, consistently allowing secure key generation at lower detection efficiencies. At $\nu = 1$, we recover the earlier thresholds from Fig.~\ref{Key rate Detection eff}, confirming consistency between the two analyses.
 We emphasize that the analysis extends beyond the ideal-visibility limit. Figure~\ref{eta_vs_nu} shows that the detection-efficiency requirements become increasingly stringent as the visibility decreases in the experimentally relevant nonideal regime $\nu<1$. For example, at $\nu \approx 0.95$ the minimum detection efficiency is approximately $\eta_A \approx 78\%$ with post-selection and $\eta_A \approx 86\%$ without post-selection, increasing to $\eta_A \approx 82\%$ and $\eta_A \approx 90\%$, respectively, at $\nu \approx 0.90$. Consistently, Fig.~\ref{Key rate Detection eff} indicates that a positive secret key rate is achievable only for $\eta_A \gtrsim 74.5\%$ with post-selection and $\eta_A \gtrsim 82.7\%$ without post-selection, corresponding to the high-visibility regime. These results highlight the experimentally relevant trade-off between source visibility and detector efficiency in non-ideal operating conditions.

When compared with other known QKD protocols, our approach remains competitive. In entanglement-based protocols with a trusted Bob, critical detection efficiencies are around $89.6\%$ without postselection and $83.3\%$ with postselection~\cite{Branciard2012, Masanes2011}. In the DI setting under collective attacks, the requirements are even stricter $92.3\%$ without postselection and $88.9\%$ with it~\cite{Pironio2009}. The original one-sided DI-QKD protocol proposed by Branciard \textit{et al.} achieves lower thresholds of $78\%$ (without postselection) and $65.9\%$ (with postselection) using a two-setting BBM92-like scheme and an entropic uncertainty-based proof~\cite{Branciard2012}. Masini \emph{et al.}~\cite{masini20241sDIQKD} further reduce the required detection efficiency on the untrusted side to 
50.1\% by retaining nondetection events as an explicit third outcome and bounding Eve’s information numerically using SDP-based entropy-accumulation techniques. Although this three-outcome treatment avoids post-selection and enables security against coherent attacks, it significantly weakens the effective correlations, leading to reduced key rates near the threshold and security bounds that are purely numerical.

\begin{figure}[h!]
    \centering
    \includegraphics[width=\linewidth]{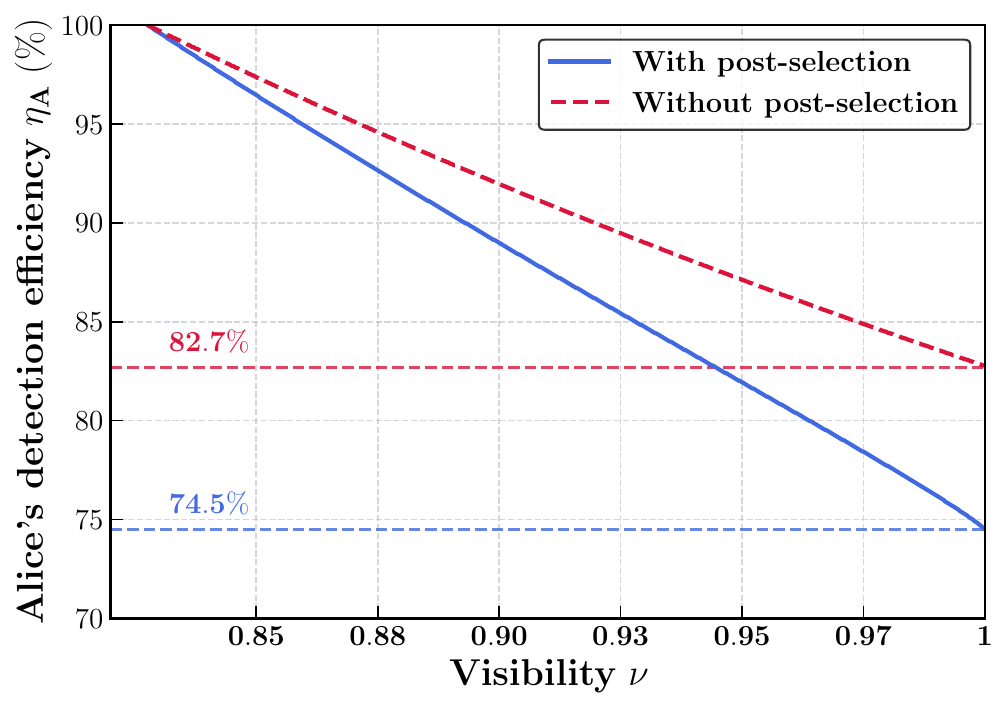}
    \caption{\footnotesize 
Threshold detection efficiency of Alice's device $\eta_A$ (in \%) required for a positive secret key rate $r^{1sDI}> 0$ as a function of the source visibility $\nu$. The non-postselected strategy is denoted by the red dashed line, while the postselection case is represented by the solid blue line. 
}
    \label{eta_vs_nu}
\end{figure}

Notably, the critical detection efficiency in steering-based protocols depends on the chosen steering inequality. While we use a fixed bound, experimental studies suggest that the threshold can be adapted based on observed efficiencies~\cite{BennetExptSteering2012}, potentially pushing the security threshold even lower.

\section{Salient Features and Outlook}\label{Salient}

We have presented a security framework for 1sDI-QKD based on the violation of the CJWR steering inequality, offering a practically motivated middle ground between DI and DD quantum cryptographic protocols. By assuming trust only in Bob’s measurement device and treating Alice’s device as completely untrusted, our approach aligns well with realistic scenarios, such as asymmetric user-server QKD architectures.

A key strength of this work is the explicit analytical characterization of the asymptotic secret key rate under optimal collective attacks. By relating the experimentally observed CJWR steering violation $\mathcal{F}_3$ to Eve’s Holevo information within an effective two-qubit Bell-diagonal description, we derive a closed-form Devetak–Winter key-rate expression that can be evaluated directly from measured statistics. In contrast to QBER-based approaches, our protocol follows an E91-style logic, analogous to device-independent QKD based on Bell-inequality violations~\cite{Acin2007}, in which Eve’s information is quantified directly through the degree of CJWR steering violation.

Furthermore, our protocol remains robust against depolarizing noise, tolerating up to $8.62\%$ QBER by relying solely on the observed correlations $(Q, \mathcal{F}_3)$.
For comparison, DI-QKD protocols using three-setting Bell inequalities tolerate up to $Q_c = 7.5\%$~\cite{Masanes2011}, while those based on asymmetric Bell inequalities reach $Q_c = 8.34\%$~\cite{Woodhead2021}. The resulting performance approaches that of fully trusted DD schemes, yet preserves significant device-independence, reinforcing the value of 1sDI-QKD in practical implementations. Our protocol also shows favourable thresholds under detection inefficiency, remaining secure with efficiencies as low as $74.5\%$ under post-selection. This compares favourably to DI-QKD thresholds, which often demand efficiencies greater than $87\%$, as demonstrated in recent photonic DI-QKD (with noisy preprocessing) implementation~\cite{LiuDI-QKDExpt}.

 While the present analysis is restricted to asymptotic security against collective attacks, such attacks capture the optimal single-round eavesdropping strategy and have historically played a central role in the development of device-independent security proofs~\cite{Acin2007}. In standard, device-dependent scenarios, collective-attack bounds can often be lifted to coherent-attack security via de Finetti arguments; however, in DI and 1sDI settings, this route is considerably more involved, as known quantum de Finetti theorems rely on assumptions such as bounded local dimension that are not available when devices are uncharacterized~\cite{ChristandlTonerDeFinetti2009}. A full extension to composable security against general coherent attacks, including finite-size effects, therefore requires alternative tools, most notably entropy accumulation techniques capable of handling general sequential and memory-dependent strategies~\cite{FriedmanEAT2018,Dupuis2020EAT}. In this context, we expect that the explicit single-round bounds on Eve’s information derived here will provide a natural and valuable building block for such future extensions beyond the collective-attack regime.

This work's analytical framework opens up a number of promising directions for the advancement of 1sDI-QKD in practice. Incorporating noisy preprocessing~\cite{HoNoisyPreProcessQKD2020, Sekatski2021} is a crucial extension that could lower the critical detection-efficiency threshold and increase robustness against experimental imperfections. In addition, experimental studies ~\cite{BennetExptSteering2012} show that linear steering inequalities can be modified to account for observed losses, as discussed in Section~\ref{DetectionEfficiency}. This suggests that our detection thresholds could be made more tolerant in practical implementations.
 Lastly, investigating alternative steering inequalities with improved loss and noise resilience~\cite{Pramanik2014, sumuncert} may increase the applicability of our method across various quantum network architectures. Such directions collectively aim to strengthen the viability of 1sDI-QKD as a secure and scalable solution for near-term quantum communication.

%




\begin{thebibliography}{71}%
	\makeatletter
	\providecommand \@ifxundefined [1]{%
		\@ifx{#1\undefined}
	}%
	\providecommand \@ifnum [1]{%
		\ifnum #1\expandafter \@firstoftwo
		\else \expandafter \@secondoftwo
		\fi
	}%
	\providecommand \@ifx [1]{%
		\ifx #1\expandafter \@firstoftwo
		\else \expandafter \@secondoftwo
		\fi
	}%
	\providecommand \natexlab [1]{#1}%
	\providecommand \enquote  [1]{``#1''}%
	\providecommand \bibnamefont  [1]{#1}%
	\providecommand \bibfnamefont [1]{#1}%
	\providecommand \citenamefont [1]{#1}%
	\providecommand \href@noop [0]{\@secondoftwo}%
	\providecommand \href [0]{\begingroup \@sanitize@url \@href}%
	\providecommand \@href[1]{\@@startlink{#1}\@@href}%
	\providecommand \@@href[1]{\endgroup#1\@@endlink}%
	\providecommand \@sanitize@url [0]{\catcode `\\12\catcode `\$12\catcode
		`\&12\catcode `\#12\catcode `\^12\catcode `\_12\catcode `\%12\relax}%
	\providecommand \@@startlink[1]{}%
	\providecommand \@@endlink[0]{}%
	\providecommand \url  [0]{\begingroup\@sanitize@url \@url }%
	\providecommand \@url [1]{\endgroup\@href {#1}{\urlprefix }}%
	\providecommand \urlprefix  [0]{URL }%
	\providecommand \Eprint [0]{\href }%
	\providecommand \doibase [0]{https://doi.org/}%
	\providecommand \selectlanguage [0]{\@gobble}%
	\providecommand \bibinfo  [0]{\@secondoftwo}%
	\providecommand \bibfield  [0]{\@secondoftwo}%
	\providecommand \translation [1]{[#1]}%
	\providecommand \BibitemOpen [0]{}%
	\providecommand \bibitemStop [0]{}%
	\providecommand \bibitemNoStop [0]{.\EOS\space}%
	\providecommand \EOS [0]{\spacefactor3000\relax}%
	\providecommand \BibitemShut  [1]{\csname bibitem#1\endcsname}%
	\let\auto@bib@innerbib\@empty
	\bibitem [{\citenamefont {Shannon}(1949)}]{Shannon1949}%
	\BibitemOpen
	\bibfield  {author} {\bibinfo {author} {\bibfnamefont {C.~E.}\ \bibnamefont
			{Shannon}},\ }\bibfield  {title} {\bibinfo {title} {Communication theory of
			secrecy systems},\ }\href
	{https://doi.org/10.1002/j.1538-7305.1949.tb00928.x} {\bibfield  {journal}
		{\bibinfo  {journal} {The Bell System Technical Journal}\ }\textbf {\bibinfo
			{volume} {28}},\ \bibinfo {pages} {656} (\bibinfo {year} {1949})}\BibitemShut
	{NoStop}%
	\bibitem [{\citenamefont {Rivest}\ \emph {et~al.}(1978)\citenamefont {Rivest},
		\citenamefont {Shamir},\ and\ \citenamefont {Adleman}}]{RSA1978}%
	\BibitemOpen
	\bibfield  {author} {\bibinfo {author} {\bibfnamefont {R.~L.}\ \bibnamefont
			{Rivest}}, \bibinfo {author} {\bibfnamefont {A.}~\bibnamefont {Shamir}},\
		and\ \bibinfo {author} {\bibfnamefont {L.}~\bibnamefont {Adleman}},\
	}\bibfield  {title} {\bibinfo {title} {A method for obtaining digital
			signatures and public-key cryptosystems},\ }\href
	{https://doi.org/10.1145/359340.359342} {\bibfield  {journal} {\bibinfo
			{journal} {Communications of the ACM}\ }\textbf {\bibinfo {volume} {21}},\
		\bibinfo {pages} {120–126} (\bibinfo {year} {1978})}\BibitemShut {NoStop}%
	\bibitem [{\citenamefont {Bennett}\ and\ \citenamefont
		{Brassard}(2014)}]{BB84}%
	\BibitemOpen
	\bibfield  {author} {\bibinfo {author} {\bibfnamefont {C.~H.}\ \bibnamefont
			{Bennett}}\ and\ \bibinfo {author} {\bibfnamefont {G.}~\bibnamefont
			{Brassard}},\ }\bibfield  {title} {\bibinfo {title} {Quantum cryptography:
			Public key distribution and coin tossing},\ }\href
	{https://doi.org/https://doi.org/10.1016/j.tcs.2014.05.025} {\bibfield
		{journal} {\bibinfo  {journal} {Theoretical Computer Science}\ }\textbf
		{\bibinfo {volume} {560}},\ \bibinfo {pages} {7} (\bibinfo {year} {2014})},\
	\bibinfo {note} {theoretical Aspects of Quantum Cryptography – celebrating
		30 years of BB84}\BibitemShut {NoStop}%
	\bibitem [{\citenamefont {Ekert}(1991)}]{Ekert1991}%
	\BibitemOpen
	\bibfield  {author} {\bibinfo {author} {\bibfnamefont {A.~K.}\ \bibnamefont
			{Ekert}},\ }\bibfield  {title} {\bibinfo {title} {Quantum cryptography based
			on bell's theorem},\ }\href {https://doi.org/10.1103/PhysRevLett.67.661}
	{\bibfield  {journal} {\bibinfo  {journal} {Phys. Rev. Lett.}\ }\textbf
		{\bibinfo {volume} {67}},\ \bibinfo {pages} {661} (\bibinfo {year}
		{1991})}\BibitemShut {NoStop}%
	\bibitem [{\citenamefont {Bell}(1964)}]{Bell1964}%
	\BibitemOpen
	\bibfield  {author} {\bibinfo {author} {\bibfnamefont {J.~S.}\ \bibnamefont
			{Bell}},\ }\bibfield  {title} {\bibinfo {title} {On the einstein podolsky
			rosen paradox},\ }\href {https://doi.org/10.1103/PhysicsPhysiqueFizika.1.195}
	{\bibfield  {journal} {\bibinfo  {journal} {Physics Physique Fizika}\
		}\textbf {\bibinfo {volume} {1}},\ \bibinfo {pages} {195} (\bibinfo {year}
		{1964})}\BibitemShut {NoStop}%
	\bibitem [{\citenamefont {Clauser}\ \emph {et~al.}(1969)\citenamefont
		{Clauser}, \citenamefont {Horne}, \citenamefont {Shimony},\ and\
		\citenamefont {Holt}}]{Clauser1969}%
	\BibitemOpen
	\bibfield  {author} {\bibinfo {author} {\bibfnamefont {J.~F.}\ \bibnamefont
			{Clauser}}, \bibinfo {author} {\bibfnamefont {M.~A.}\ \bibnamefont {Horne}},
		\bibinfo {author} {\bibfnamefont {A.}~\bibnamefont {Shimony}},\ and\ \bibinfo
		{author} {\bibfnamefont {R.~A.}\ \bibnamefont {Holt}},\ }\bibfield  {title}
	{\bibinfo {title} {Proposed experiment to test local hidden-variable
			theories},\ }\href {https://doi.org/10.1103/physrevlett.23.880} {\bibfield
		{journal} {\bibinfo  {journal} {Phys. Rev. Lett.}\ }\textbf {\bibinfo
			{volume} {23}},\ \bibinfo {pages} {880–884} (\bibinfo {year}
		{1969})}\BibitemShut {NoStop}%
	\bibitem [{\citenamefont {Bennett}\ \emph {et~al.}(1992)\citenamefont
		{Bennett}, \citenamefont {Brassard},\ and\ \citenamefont {Mermin}}]{BBM92}%
	\BibitemOpen
	\bibfield  {author} {\bibinfo {author} {\bibfnamefont {C.~H.}\ \bibnamefont
			{Bennett}}, \bibinfo {author} {\bibfnamefont {G.}~\bibnamefont {Brassard}},\
		and\ \bibinfo {author} {\bibfnamefont {N.~D.}\ \bibnamefont {Mermin}},\
	}\bibfield  {title} {\bibinfo {title} {Quantum cryptography without bell's
			theorem},\ }\href {https://doi.org/10.1103/PhysRevLett.68.557} {\bibfield
		{journal} {\bibinfo  {journal} {Phys. Rev. Lett.}\ }\textbf {\bibinfo
			{volume} {68}},\ \bibinfo {pages} {557} (\bibinfo {year} {1992})}\BibitemShut
	{NoStop}%
	\bibitem [{\citenamefont {Horodecki}\ \emph {et~al.}(2009)\citenamefont
		{Horodecki}, \citenamefont {Horodecki}, \citenamefont {Horodecki},\ and\
		\citenamefont {Horodecki}}]{HorodeckiQentanglement2009}%
	\BibitemOpen
	\bibfield  {author} {\bibinfo {author} {\bibfnamefont {R.}~\bibnamefont
			{Horodecki}}, \bibinfo {author} {\bibfnamefont {P.}~\bibnamefont
			{Horodecki}}, \bibinfo {author} {\bibfnamefont {M.}~\bibnamefont
			{Horodecki}},\ and\ \bibinfo {author} {\bibfnamefont {K.}~\bibnamefont
			{Horodecki}},\ }\bibfield  {title} {\bibinfo {title} {Quantum entanglement},\
	}\href {https://doi.org/10.1103/RevModPhys.81.865} {\bibfield  {journal}
		{\bibinfo  {journal} {Rev. Mod. Phys.}\ }\textbf {\bibinfo {volume} {81}},\
		\bibinfo {pages} {865} (\bibinfo {year} {2009})}\BibitemShut {NoStop}%
	\bibitem [{\citenamefont {Lydersen}\ \emph {et~al.}(2010)\citenamefont
		{Lydersen}, \citenamefont {Wiechers}, \citenamefont {Wittmann}, \citenamefont
		{Elser}, \citenamefont {Skaar},\ and\ \citenamefont
		{Makarov}}]{Lydersen2010}%
	\BibitemOpen
	\bibfield  {author} {\bibinfo {author} {\bibfnamefont {L.}~\bibnamefont
			{Lydersen}}, \bibinfo {author} {\bibfnamefont {C.}~\bibnamefont {Wiechers}},
		\bibinfo {author} {\bibfnamefont {C.}~\bibnamefont {Wittmann}}, \bibinfo
		{author} {\bibfnamefont {D.}~\bibnamefont {Elser}}, \bibinfo {author}
		{\bibfnamefont {J.}~\bibnamefont {Skaar}},\ and\ \bibinfo {author}
		{\bibfnamefont {V.}~\bibnamefont {Makarov}},\ }\bibfield  {title} {\bibinfo
		{title} {Hacking commercial quantum cryptography systems by tailored bright
			illumination},\ }\href {https://doi.org/10.1038/nphoton.2010.214} {\bibfield
		{journal} {\bibinfo  {journal} {Nat. Photonics}\ }\textbf {\bibinfo {volume}
			{4}},\ \bibinfo {pages} {686} (\bibinfo {year} {2010})}\BibitemShut {NoStop}%
	\bibitem [{\citenamefont {Scarani}\ \emph {et~al.}(2009)\citenamefont
		{Scarani}, \citenamefont {Bechmann-Pasquinucci}, \citenamefont {Cerf},
		\citenamefont {Du\ifmmode~\check{s}\else \v{s}\fi{}ek}, \citenamefont
		{L\"utkenhaus},\ and\ \citenamefont {Peev}}]{Scarani2009}%
	\BibitemOpen
	\bibfield  {author} {\bibinfo {author} {\bibfnamefont {V.}~\bibnamefont
			{Scarani}}, \bibinfo {author} {\bibfnamefont {H.}~\bibnamefont
			{Bechmann-Pasquinucci}}, \bibinfo {author} {\bibfnamefont {N.~J.}\
			\bibnamefont {Cerf}}, \bibinfo {author} {\bibfnamefont {M.}~\bibnamefont
			{Du\ifmmode~\check{s}\else \v{s}\fi{}ek}}, \bibinfo {author} {\bibfnamefont
			{N.}~\bibnamefont {L\"utkenhaus}},\ and\ \bibinfo {author} {\bibfnamefont
			{M.}~\bibnamefont {Peev}},\ }\bibfield  {title} {\bibinfo {title} {The
			security of practical quantum key distribution},\ }\href
	{https://doi.org/10.1103/RevModPhys.81.1301} {\bibfield  {journal} {\bibinfo
			{journal} {Rev. Mod. Phys.}\ }\textbf {\bibinfo {volume} {81}},\ \bibinfo
		{pages} {1301} (\bibinfo {year} {2009})}\BibitemShut {NoStop}%
	\bibitem [{\citenamefont {Ekert}\ \emph {et~al.}(1994)\citenamefont {Ekert},
		\citenamefont {Huttner}, \citenamefont {Palma},\ and\ \citenamefont
		{Peres}}]{Ekert1994}%
	\BibitemOpen
	\bibfield  {author} {\bibinfo {author} {\bibfnamefont {A.~K.}\ \bibnamefont
			{Ekert}}, \bibinfo {author} {\bibfnamefont {B.}~\bibnamefont {Huttner}},
		\bibinfo {author} {\bibfnamefont {G.~M.}\ \bibnamefont {Palma}},\ and\
		\bibinfo {author} {\bibfnamefont {A.}~\bibnamefont {Peres}},\ }\bibfield
	{title} {\bibinfo {title} {Eavesdropping on quantum-cryptographical
			systems},\ }\href {https://doi.org/10.1103/physreva.50.1047} {\bibfield
		{journal} {\bibinfo  {journal} {Phys. Rev. A}\ }\textbf {\bibinfo {volume}
			{50}},\ \bibinfo {pages} {1047–1056} (\bibinfo {year} {1994})}\BibitemShut
	{NoStop}%
	\bibitem [{\citenamefont {Slutsky}\ \emph {et~al.}(1998)\citenamefont
		{Slutsky}, \citenamefont {Rao}, \citenamefont {Sun},\ and\ \citenamefont
		{Fainman}}]{Slutsky1998}%
	\BibitemOpen
	\bibfield  {author} {\bibinfo {author} {\bibfnamefont {B.~A.}\ \bibnamefont
			{Slutsky}}, \bibinfo {author} {\bibfnamefont {R.}~\bibnamefont {Rao}},
		\bibinfo {author} {\bibfnamefont {P.-C.}\ \bibnamefont {Sun}},\ and\ \bibinfo
		{author} {\bibfnamefont {Y.}~\bibnamefont {Fainman}},\ }\bibfield  {title}
	{\bibinfo {title} {Security of quantum cryptography against individual
			attacks},\ }\href {https://doi.org/10.1103/physreva.57.2383} {\bibfield
		{journal} {\bibinfo  {journal} {Phys. Rev. A}\ }\textbf {\bibinfo {volume}
			{57}},\ \bibinfo {pages} {2383–2398} (\bibinfo {year} {1998})}\BibitemShut
	{NoStop}%
	\bibitem [{\citenamefont {L\"utkenhaus}(2000)}]{LutkenhausIndAttack2000}%
	\BibitemOpen
	\bibfield  {author} {\bibinfo {author} {\bibfnamefont {N.}~\bibnamefont
			{L\"utkenhaus}},\ }\bibfield  {title} {\bibinfo {title} {Security against
			individual attacks for realistic quantum key distribution},\ }\href
	{https://doi.org/10.1103/PhysRevA.61.052304} {\bibfield  {journal} {\bibinfo
			{journal} {Phys. Rev. A}\ }\textbf {\bibinfo {volume} {61}},\ \bibinfo
		{pages} {052304} (\bibinfo {year} {2000})}\BibitemShut {NoStop}%
	\bibitem [{\citenamefont {Curty}\ and\ \citenamefont
		{L\"utkenhaus}(2005)}]{Curty2005}%
	\BibitemOpen
	\bibfield  {author} {\bibinfo {author} {\bibfnamefont {M.}~\bibnamefont
			{Curty}}\ and\ \bibinfo {author} {\bibfnamefont {N.}~\bibnamefont
			{L\"utkenhaus}},\ }\bibfield  {title} {\bibinfo {title} {Intercept-resend
			attacks in the bennett-brassard 1984 quantum-key-distribution protocol with
			weak coherent pulses},\ }\href {https://doi.org/10.1103/PhysRevA.71.062301}
	{\bibfield  {journal} {\bibinfo  {journal} {Phys. Rev. A}\ }\textbf {\bibinfo
			{volume} {71}},\ \bibinfo {pages} {062301} (\bibinfo {year}
		{2005})}\BibitemShut {NoStop}%
	\bibitem [{\citenamefont {Roy}\ \emph {et~al.}(2024)\citenamefont {Roy},
		\citenamefont {Sasmal}, \citenamefont {Bera}, \citenamefont {Gupta},
		\citenamefont {Roy},\ and\ \citenamefont {Majumdar}}]{roy2024sequential}%
	\BibitemOpen
	\bibfield  {author} {\bibinfo {author} {\bibfnamefont {P.}~\bibnamefont
			{Roy}}, \bibinfo {author} {\bibfnamefont {S.}~\bibnamefont {Sasmal}},
		\bibinfo {author} {\bibfnamefont {S.}~\bibnamefont {Bera}}, \bibinfo {author}
		{\bibfnamefont {S.}~\bibnamefont {Gupta}}, \bibinfo {author} {\bibfnamefont
			{A.}~\bibnamefont {Roy}},\ and\ \bibinfo {author} {\bibfnamefont {A.~S.}\
			\bibnamefont {Majumdar}},\ }\bibfield  {title} {\bibinfo {title} {Security of
			device-independent quantum key distribution under sequential attack},\ }\href
	{https://doi.org/10.48550/arXiv.2411.16822} {\bibfield  {journal} {\bibinfo
			{journal} {arXiv preprint arXiv:2411.16822}\ } (\bibinfo {year}
		{2024})}\BibitemShut {NoStop}%
	\bibitem [{\citenamefont {Ac\'{\i}n}\ \emph {et~al.}(2007)\citenamefont
		{Ac\'{\i}n}, \citenamefont {Brunner}, \citenamefont {Gisin}, \citenamefont
		{Massar}, \citenamefont {Pironio},\ and\ \citenamefont {Scarani}}]{Acin2007}%
	\BibitemOpen
	\bibfield  {author} {\bibinfo {author} {\bibfnamefont {A.}~\bibnamefont
			{Ac\'{\i}n}}, \bibinfo {author} {\bibfnamefont {N.}~\bibnamefont {Brunner}},
		\bibinfo {author} {\bibfnamefont {N.}~\bibnamefont {Gisin}}, \bibinfo
		{author} {\bibfnamefont {S.}~\bibnamefont {Massar}}, \bibinfo {author}
		{\bibfnamefont {S.}~\bibnamefont {Pironio}},\ and\ \bibinfo {author}
		{\bibfnamefont {V.}~\bibnamefont {Scarani}},\ }\bibfield  {title} {\bibinfo
		{title} {Device-independent security of quantum cryptography against
			collective attacks},\ }\href {https://doi.org/10.1103/PhysRevLett.98.230501}
	{\bibfield  {journal} {\bibinfo  {journal} {Phys. Rev. Lett.}\ }\textbf
		{\bibinfo {volume} {98}},\ \bibinfo {pages} {230501} (\bibinfo {year}
		{2007})}\BibitemShut {NoStop}%
	\bibitem [{\citenamefont {Shor}\ and\ \citenamefont
		{Preskill}(2000)}]{Shor2000}%
	\BibitemOpen
	\bibfield  {author} {\bibinfo {author} {\bibfnamefont {P.~W.}\ \bibnamefont
			{Shor}}\ and\ \bibinfo {author} {\bibfnamefont {J.}~\bibnamefont
			{Preskill}},\ }\bibfield  {title} {\bibinfo {title} {Simple proof of security
			of the bb84 quantum key distribution protocol},\ }\href
	{https://doi.org/10.1103/physrevlett.85.441} {\bibfield  {journal} {\bibinfo
			{journal} {Phys. Rev. Lett}\ }\textbf {\bibinfo {volume} {85}},\ \bibinfo
		{pages} {441–444} (\bibinfo {year} {2000})}\BibitemShut {NoStop}%
	\bibitem [{\citenamefont {Masanes}\ \emph {et~al.}(2011)\citenamefont
		{Masanes}, \citenamefont {Pironio},\ and\ \citenamefont
		{Ac{\'i}n}}]{Masanes2011}%
	\BibitemOpen
	\bibfield  {author} {\bibinfo {author} {\bibfnamefont {L.}~\bibnamefont
			{Masanes}}, \bibinfo {author} {\bibfnamefont {S.}~\bibnamefont {Pironio}},\
		and\ \bibinfo {author} {\bibfnamefont {A.}~\bibnamefont {Ac{\'i}n}},\
	}\bibfield  {title} {\bibinfo {title} {Secure device-independent quantum key
			distribution with causally independent measurement devices},\ }\href
	{https://doi.org/10.1038/ncomms1244} {\bibfield  {journal} {\bibinfo
			{journal} {Nature Communications}\ }\textbf {\bibinfo {volume} {2}},\
		\bibinfo {pages} {238} (\bibinfo {year} {2011})}\BibitemShut {NoStop}%
	\bibitem [{\citenamefont {Ac\'{\i}n}\ \emph {et~al.}(2006)\citenamefont
		{Ac\'{\i}n}, \citenamefont {Gisin},\ and\ \citenamefont
		{Masanes}}]{Acin2006}%
	\BibitemOpen
	\bibfield  {author} {\bibinfo {author} {\bibfnamefont {A.}~\bibnamefont
			{Ac\'{\i}n}}, \bibinfo {author} {\bibfnamefont {N.}~\bibnamefont {Gisin}},\
		and\ \bibinfo {author} {\bibfnamefont {L.}~\bibnamefont {Masanes}},\
	}\bibfield  {title} {\bibinfo {title} {From bell's theorem to secure quantum
			key distribution},\ }\href {https://doi.org/10.1103/PhysRevLett.97.120405}
	{\bibfield  {journal} {\bibinfo  {journal} {Phys. Rev. Lett.}\ }\textbf
		{\bibinfo {volume} {97}},\ \bibinfo {pages} {120405} (\bibinfo {year}
		{2006})}\BibitemShut {NoStop}%
	\bibitem [{\citenamefont {Lo}\ and\ \citenamefont {Chau}(1999)}]{LoChau1999}%
	\BibitemOpen
	\bibfield  {author} {\bibinfo {author} {\bibfnamefont {H.-K.}\ \bibnamefont
			{Lo}}\ and\ \bibinfo {author} {\bibfnamefont {H.~F.}\ \bibnamefont {Chau}},\
	}\bibfield  {title} {\bibinfo {title} {Unconditional security of quantum key
			distribution over arbitrarily long distances},\ }\href
	{https://doi.org/10.1126/science.283.5410.2050} {\bibfield  {journal}
		{\bibinfo  {journal} {Science}\ }\textbf {\bibinfo {volume} {283}},\ \bibinfo
		{pages} {2050} (\bibinfo {year} {1999})}\BibitemShut {NoStop}%
	\bibitem [{\citenamefont {Vazirani}\ and\ \citenamefont
		{Vidick}(2014)}]{Vazirani14}%
	\BibitemOpen
	\bibfield  {author} {\bibinfo {author} {\bibfnamefont {U.}~\bibnamefont
			{Vazirani}}\ and\ \bibinfo {author} {\bibfnamefont {T.}~\bibnamefont
			{Vidick}},\ }\bibfield  {title} {\bibinfo {title} {Fully device-independent
			quantum key distribution},\ }\href
	{https://doi.org/10.1103/PhysRevLett.113.140501} {\bibfield  {journal}
		{\bibinfo  {journal} {Phys. Rev. Lett.}\ }\textbf {\bibinfo {volume} {113}},\
		\bibinfo {pages} {140501} (\bibinfo {year} {2014})}\BibitemShut {NoStop}%
	\bibitem [{\citenamefont {Woodhead}\ \emph {et~al.}(2021)\citenamefont
		{Woodhead}, \citenamefont {Acín},\ and\ \citenamefont
		{Pironio}}]{Woodhead2021}%
	\BibitemOpen
	\bibfield  {author} {\bibinfo {author} {\bibfnamefont {E.}~\bibnamefont
			{Woodhead}}, \bibinfo {author} {\bibfnamefont {A.}~\bibnamefont {Acín}},\
		and\ \bibinfo {author} {\bibfnamefont {S.}~\bibnamefont {Pironio}},\
	}\bibfield  {title} {\bibinfo {title} {Device-independent quantum key
			distribution with asymmetric chsh inequalities},\ }\href
	{https://doi.org/10.22331/q-2021-04-26-443} {\bibfield  {journal} {\bibinfo
			{journal} {Quantum}\ }\textbf {\bibinfo {volume} {5}},\ \bibinfo {pages}
		{443} (\bibinfo {year} {2021})}\BibitemShut {NoStop}%
	\bibitem [{\citenamefont {Bera}\ \emph {et~al.}(2023)\citenamefont {Bera},
		\citenamefont {Gupta},\ and\ \citenamefont {Majumdar}}]{BeraRandomQKD2023}%
	\BibitemOpen
	\bibfield  {author} {\bibinfo {author} {\bibfnamefont {S.}~\bibnamefont
			{Bera}}, \bibinfo {author} {\bibfnamefont {S.}~\bibnamefont {Gupta}},\ and\
		\bibinfo {author} {\bibfnamefont {A.~S.}\ \bibnamefont {Majumdar}},\
	}\bibfield  {title} {\bibinfo {title} {Device-independent quantum key
			distribution using random quantum states},\ }\href
	{https://doi.org/10.1007/s11128-023-03852-2} {\bibfield  {journal} {\bibinfo
			{journal} {Quantum Information Processing}\ }\textbf {\bibinfo {volume}
			{22}},\ \bibinfo {pages} {109} (\bibinfo {year} {2023})}\BibitemShut
	{NoStop}%
	\bibitem [{\citenamefont {Schwonnek}\ \emph {et~al.}(2021)\citenamefont
		{Schwonnek}, \citenamefont {Goh}, \citenamefont {Primaatmaja}, \citenamefont
		{Tan}, \citenamefont {Wolf}, \citenamefont {Scarani},\ and\ \citenamefont
		{Lim}}]{SchwonnekRandomBasis2021}%
	\BibitemOpen
	\bibfield  {author} {\bibinfo {author} {\bibfnamefont {R.}~\bibnamefont
			{Schwonnek}}, \bibinfo {author} {\bibfnamefont {K.~T.}\ \bibnamefont {Goh}},
		\bibinfo {author} {\bibfnamefont {I.~W.}\ \bibnamefont {Primaatmaja}},
		\bibinfo {author} {\bibfnamefont {E.~Y.-Z.}\ \bibnamefont {Tan}}, \bibinfo
		{author} {\bibfnamefont {R.}~\bibnamefont {Wolf}}, \bibinfo {author}
		{\bibfnamefont {V.}~\bibnamefont {Scarani}},\ and\ \bibinfo {author}
		{\bibfnamefont {C.~C.-W.}\ \bibnamefont {Lim}},\ }\bibfield  {title}
	{\bibinfo {title} {Device-independent quantum key distribution with random
			key basis},\ }\href
	{https://doi.org/https://doi.org/10.1038/s41467-021-23147-3} {\bibfield
		{journal} {\bibinfo  {journal} {Nature communications}\ }\textbf {\bibinfo
			{volume} {12}},\ \bibinfo {pages} {2880} (\bibinfo {year}
		{2021})}\BibitemShut {NoStop}%
	\bibitem [{\citenamefont {Liu}\ \emph {et~al.}(2022)\citenamefont {Liu},
		\citenamefont {Zhang}, \citenamefont {Zhen}, \citenamefont {Li},
		\citenamefont {Liu}, \citenamefont {Fan}, \citenamefont {Xu}, \citenamefont
		{Zhang},\ and\ \citenamefont {Pan}}]{LiuDI-QKDExpt}%
	\BibitemOpen
	\bibfield  {author} {\bibinfo {author} {\bibfnamefont {W.-Z.}\ \bibnamefont
			{Liu}}, \bibinfo {author} {\bibfnamefont {Y.-Z.}\ \bibnamefont {Zhang}},
		\bibinfo {author} {\bibfnamefont {Y.-Z.}\ \bibnamefont {Zhen}}, \bibinfo
		{author} {\bibfnamefont {M.-H.}\ \bibnamefont {Li}}, \bibinfo {author}
		{\bibfnamefont {Y.}~\bibnamefont {Liu}}, \bibinfo {author} {\bibfnamefont
			{J.}~\bibnamefont {Fan}}, \bibinfo {author} {\bibfnamefont {F.}~\bibnamefont
			{Xu}}, \bibinfo {author} {\bibfnamefont {Q.}~\bibnamefont {Zhang}},\ and\
		\bibinfo {author} {\bibfnamefont {J.-W.}\ \bibnamefont {Pan}},\ }\bibfield
	{title} {\bibinfo {title} {Toward a photonic demonstration of
			device-independent quantum key distribution},\ }\href
	{https://doi.org/10.1103/PhysRevLett.129.050502} {\bibfield  {journal}
		{\bibinfo  {journal} {Phys. Rev. Lett.}\ }\textbf {\bibinfo {volume} {129}},\
		\bibinfo {pages} {050502} (\bibinfo {year} {2022})}\BibitemShut {NoStop}%
	\bibitem [{\citenamefont {Zhang}\ \emph {et~al.}(2022)\citenamefont {Zhang},
		\citenamefont {van Leent}, \citenamefont {Redeker},\ and\ \citenamefont
		{et~al.}}]{Zhang2022DIQKD}%
	\BibitemOpen
	\bibfield  {author} {\bibinfo {author} {\bibfnamefont {W.}~\bibnamefont
			{Zhang}}, \bibinfo {author} {\bibfnamefont {T.}~\bibnamefont {van Leent}},
		\bibinfo {author} {\bibfnamefont {K.}~\bibnamefont {Redeker}},\ and\ \bibinfo
		{author} {\bibnamefont {et~al.}},\ }\bibfield  {title} {\bibinfo {title} {A
			device-independent quantum key distribution system for distant users},\
	}\href {https://doi.org/10.1038/s41586-022-04891-y} {\bibfield  {journal}
		{\bibinfo  {journal} {Nature}\ }\textbf {\bibinfo {volume} {607}},\ \bibinfo
		{pages} {687} (\bibinfo {year} {2022})}\BibitemShut {NoStop}%
	\bibitem [{\citenamefont {Pironio}\ \emph {et~al.}(2009)\citenamefont
		{Pironio}, \citenamefont {Acín}, \citenamefont {Brunner}, \citenamefont
		{Gisin}, \citenamefont {Massar},\ and\ \citenamefont
		{Scarani}}]{Pironio2009}%
	\BibitemOpen
	\bibfield  {author} {\bibinfo {author} {\bibfnamefont {S.}~\bibnamefont
			{Pironio}}, \bibinfo {author} {\bibfnamefont {A.}~\bibnamefont {Acín}},
		\bibinfo {author} {\bibfnamefont {N.}~\bibnamefont {Brunner}}, \bibinfo
		{author} {\bibfnamefont {N.}~\bibnamefont {Gisin}}, \bibinfo {author}
		{\bibfnamefont {S.}~\bibnamefont {Massar}},\ and\ \bibinfo {author}
		{\bibfnamefont {V.}~\bibnamefont {Scarani}},\ }\bibfield  {title} {\bibinfo
		{title} {Device-independent quantum key distribution secure against
			collective attacks},\ }\href {https://doi.org/10.1088/1367-2630/11/4/045021}
	{\bibfield  {journal} {\bibinfo  {journal} {New Journal of Physics}\ }\textbf
		{\bibinfo {volume} {11}},\ \bibinfo {pages} {045021} (\bibinfo {year}
		{2009})}\BibitemShut {NoStop}%
	\bibitem [{\citenamefont {Shalm}\ \emph {et~al.}(2015)\citenamefont {Shalm},
		\citenamefont {Meyer-Scott}, \citenamefont {Christensen}, \citenamefont
		{Bierhorst}, \citenamefont {Wayne}, \citenamefont {Stevens}, \citenamefont
		{Gerrits}, \citenamefont {Glancy}, \citenamefont {Hamel}, \citenamefont
		{Allman}, \citenamefont {Coakley}, \citenamefont {Dyer}, \citenamefont
		{Hodge}, \citenamefont {Lita}, \citenamefont {Verma}, \citenamefont
		{Lambrocco}, \citenamefont {Tortorici}, \citenamefont {Migdall},
		\citenamefont {Zhang}, \citenamefont {Kumor}, \citenamefont {Farr},
		\citenamefont {Marsili}, \citenamefont {Shaw}, \citenamefont {Stern},
		\citenamefont {Abell\'an}, \citenamefont {Amaya}, \citenamefont {Pruneri},
		\citenamefont {Jennewein}, \citenamefont {Mitchell}, \citenamefont {Kwiat},
		\citenamefont {Bienfang}, \citenamefont {Mirin}, \citenamefont {Knill},\ and\
		\citenamefont {Nam}}]{Shalm2015Loopholefreebelltest}%
	\BibitemOpen
	\bibfield  {author} {\bibinfo {author} {\bibfnamefont {L.~K.}\ \bibnamefont
			{Shalm}}, \bibinfo {author} {\bibfnamefont {E.}~\bibnamefont {Meyer-Scott}},
		\bibinfo {author} {\bibfnamefont {B.~G.}\ \bibnamefont {Christensen}},
		\bibinfo {author} {\bibfnamefont {P.}~\bibnamefont {Bierhorst}}, \bibinfo
		{author} {\bibfnamefont {M.~A.}\ \bibnamefont {Wayne}}, \bibinfo {author}
		{\bibfnamefont {M.~J.}\ \bibnamefont {Stevens}}, \bibinfo {author}
		{\bibfnamefont {T.}~\bibnamefont {Gerrits}}, \bibinfo {author} {\bibfnamefont
			{S.}~\bibnamefont {Glancy}}, \bibinfo {author} {\bibfnamefont {D.~R.}\
			\bibnamefont {Hamel}}, \bibinfo {author} {\bibfnamefont {M.~S.}\ \bibnamefont
			{Allman}}, \bibinfo {author} {\bibfnamefont {K.~J.}\ \bibnamefont {Coakley}},
		\bibinfo {author} {\bibfnamefont {S.~D.}\ \bibnamefont {Dyer}}, \bibinfo
		{author} {\bibfnamefont {C.}~\bibnamefont {Hodge}}, \bibinfo {author}
		{\bibfnamefont {A.~E.}\ \bibnamefont {Lita}}, \bibinfo {author}
		{\bibfnamefont {V.~B.}\ \bibnamefont {Verma}}, \bibinfo {author}
		{\bibfnamefont {C.}~\bibnamefont {Lambrocco}}, \bibinfo {author}
		{\bibfnamefont {E.}~\bibnamefont {Tortorici}}, \bibinfo {author}
		{\bibfnamefont {A.~L.}\ \bibnamefont {Migdall}}, \bibinfo {author}
		{\bibfnamefont {Y.}~\bibnamefont {Zhang}}, \bibinfo {author} {\bibfnamefont
			{D.~R.}\ \bibnamefont {Kumor}}, \bibinfo {author} {\bibfnamefont {W.~H.}\
			\bibnamefont {Farr}}, \bibinfo {author} {\bibfnamefont {F.}~\bibnamefont
			{Marsili}}, \bibinfo {author} {\bibfnamefont {M.~D.}\ \bibnamefont {Shaw}},
		\bibinfo {author} {\bibfnamefont {J.~A.}\ \bibnamefont {Stern}}, \bibinfo
		{author} {\bibfnamefont {C.}~\bibnamefont {Abell\'an}}, \bibinfo {author}
		{\bibfnamefont {W.}~\bibnamefont {Amaya}}, \bibinfo {author} {\bibfnamefont
			{V.}~\bibnamefont {Pruneri}}, \bibinfo {author} {\bibfnamefont
			{T.}~\bibnamefont {Jennewein}}, \bibinfo {author} {\bibfnamefont {M.~W.}\
			\bibnamefont {Mitchell}}, \bibinfo {author} {\bibfnamefont {P.~G.}\
			\bibnamefont {Kwiat}}, \bibinfo {author} {\bibfnamefont {J.~C.}\ \bibnamefont
			{Bienfang}}, \bibinfo {author} {\bibfnamefont {R.~P.}\ \bibnamefont {Mirin}},
		\bibinfo {author} {\bibfnamefont {E.}~\bibnamefont {Knill}},\ and\ \bibinfo
		{author} {\bibfnamefont {S.~W.}\ \bibnamefont {Nam}},\ }\bibfield  {title}
	{\bibinfo {title} {Strong loophole-free test of local realism},\ }\href
	{https://doi.org/10.1103/PhysRevLett.115.250402} {\bibfield  {journal}
		{\bibinfo  {journal} {Phys. Rev. Lett.}\ }\textbf {\bibinfo {volume} {115}},\
		\bibinfo {pages} {250402} (\bibinfo {year} {2015})}\BibitemShut {NoStop}%
	\bibitem [{\citenamefont {Giustina}\ \emph {et~al.}(2015)\citenamefont
		{Giustina}, \citenamefont {Versteegh}, \citenamefont {Wengerowsky},
		\citenamefont {Handsteiner}, \citenamefont {Hochrainer}, \citenamefont
		{Phelan}, \citenamefont {Steinlechner}, \citenamefont {Kofler}, \citenamefont
		{Larsson}, \citenamefont {Abell\'an}, \citenamefont {Amaya}, \citenamefont
		{Pruneri}, \citenamefont {Mitchell}, \citenamefont {Beyer}, \citenamefont
		{Gerrits}, \citenamefont {Lita}, \citenamefont {Shalm}, \citenamefont {Nam},
		\citenamefont {Scheidl}, \citenamefont {Ursin}, \citenamefont {Wittmann},\
		and\ \citenamefont {Zeilinger}}]{Giustina2015Loopholefreebelltest}%
	\BibitemOpen
	\bibfield  {author} {\bibinfo {author} {\bibfnamefont {M.}~\bibnamefont
			{Giustina}}, \bibinfo {author} {\bibfnamefont {M.~A.~M.}\ \bibnamefont
			{Versteegh}}, \bibinfo {author} {\bibfnamefont {S.}~\bibnamefont
			{Wengerowsky}}, \bibinfo {author} {\bibfnamefont {J.}~\bibnamefont
			{Handsteiner}}, \bibinfo {author} {\bibfnamefont {A.}~\bibnamefont
			{Hochrainer}}, \bibinfo {author} {\bibfnamefont {K.}~\bibnamefont {Phelan}},
		\bibinfo {author} {\bibfnamefont {F.}~\bibnamefont {Steinlechner}}, \bibinfo
		{author} {\bibfnamefont {J.}~\bibnamefont {Kofler}}, \bibinfo {author}
		{\bibfnamefont {J.-A.}\ \bibnamefont {Larsson}}, \bibinfo {author}
		{\bibfnamefont {C.}~\bibnamefont {Abell\'an}}, \bibinfo {author}
		{\bibfnamefont {W.}~\bibnamefont {Amaya}}, \bibinfo {author} {\bibfnamefont
			{V.}~\bibnamefont {Pruneri}}, \bibinfo {author} {\bibfnamefont {M.~W.}\
			\bibnamefont {Mitchell}}, \bibinfo {author} {\bibfnamefont {J.}~\bibnamefont
			{Beyer}}, \bibinfo {author} {\bibfnamefont {T.}~\bibnamefont {Gerrits}},
		\bibinfo {author} {\bibfnamefont {A.~E.}\ \bibnamefont {Lita}}, \bibinfo
		{author} {\bibfnamefont {L.~K.}\ \bibnamefont {Shalm}}, \bibinfo {author}
		{\bibfnamefont {S.~W.}\ \bibnamefont {Nam}}, \bibinfo {author} {\bibfnamefont
			{T.}~\bibnamefont {Scheidl}}, \bibinfo {author} {\bibfnamefont
			{R.}~\bibnamefont {Ursin}}, \bibinfo {author} {\bibfnamefont
			{B.}~\bibnamefont {Wittmann}},\ and\ \bibinfo {author} {\bibfnamefont
			{A.}~\bibnamefont {Zeilinger}},\ }\bibfield  {title} {\bibinfo {title}
		{Significant-loophole-free test of bell's theorem with entangled photons},\
	}\href {https://doi.org/10.1103/PhysRevLett.115.250401} {\bibfield  {journal}
		{\bibinfo  {journal} {Phys. Rev. Lett.}\ }\textbf {\bibinfo {volume} {115}},\
		\bibinfo {pages} {250401} (\bibinfo {year} {2015})}\BibitemShut {NoStop}%
	\bibitem [{\citenamefont {Li}\ \emph {et~al.}(2018)\citenamefont {Li},
		\citenamefont {Wu}, \citenamefont {Zhang}, \citenamefont {Liu}, \citenamefont
		{Bai}, \citenamefont {Liu}, \citenamefont {Zhang}, \citenamefont {Zhao},
		\citenamefont {Li}, \citenamefont {Wang}, \citenamefont {You}, \citenamefont
		{Munro}, \citenamefont {Yin}, \citenamefont {Zhang}, \citenamefont {Peng},
		\citenamefont {Ma}, \citenamefont {Zhang}, \citenamefont {Fan},\ and\
		\citenamefont {Pan}}]{Li2018LoopholeFreebelltest}%
	\BibitemOpen
	\bibfield  {author} {\bibinfo {author} {\bibfnamefont {M.-H.}\ \bibnamefont
			{Li}}, \bibinfo {author} {\bibfnamefont {C.}~\bibnamefont {Wu}}, \bibinfo
		{author} {\bibfnamefont {Y.}~\bibnamefont {Zhang}}, \bibinfo {author}
		{\bibfnamefont {W.-Z.}\ \bibnamefont {Liu}}, \bibinfo {author} {\bibfnamefont
			{B.}~\bibnamefont {Bai}}, \bibinfo {author} {\bibfnamefont {Y.}~\bibnamefont
			{Liu}}, \bibinfo {author} {\bibfnamefont {W.}~\bibnamefont {Zhang}}, \bibinfo
		{author} {\bibfnamefont {Q.}~\bibnamefont {Zhao}}, \bibinfo {author}
		{\bibfnamefont {H.}~\bibnamefont {Li}}, \bibinfo {author} {\bibfnamefont
			{Z.}~\bibnamefont {Wang}}, \bibinfo {author} {\bibfnamefont {L.}~\bibnamefont
			{You}}, \bibinfo {author} {\bibfnamefont {W.~J.}\ \bibnamefont {Munro}},
		\bibinfo {author} {\bibfnamefont {J.}~\bibnamefont {Yin}}, \bibinfo {author}
		{\bibfnamefont {J.}~\bibnamefont {Zhang}}, \bibinfo {author} {\bibfnamefont
			{C.-Z.}\ \bibnamefont {Peng}}, \bibinfo {author} {\bibfnamefont
			{X.}~\bibnamefont {Ma}}, \bibinfo {author} {\bibfnamefont {Q.}~\bibnamefont
			{Zhang}}, \bibinfo {author} {\bibfnamefont {J.}~\bibnamefont {Fan}},\ and\
		\bibinfo {author} {\bibfnamefont {J.-W.}\ \bibnamefont {Pan}},\ }\bibfield
	{title} {\bibinfo {title} {Test of local realism into the past without
			detection and locality loopholes},\ }\href
	{https://doi.org/10.1103/PhysRevLett.121.080404} {\bibfield  {journal}
		{\bibinfo  {journal} {Phys. Rev. Lett.}\ }\textbf {\bibinfo {volume} {121}},\
		\bibinfo {pages} {080404} (\bibinfo {year} {2018})}\BibitemShut {NoStop}%
	\bibitem [{\citenamefont {Ho}\ \emph {et~al.}(2020)\citenamefont {Ho},
		\citenamefont {Sekatski}, \citenamefont {Tan}, \citenamefont {Renner},
		\citenamefont {Bancal},\ and\ \citenamefont
		{Sangouard}}]{HoNoisyPreProcessQKD2020}%
	\BibitemOpen
	\bibfield  {author} {\bibinfo {author} {\bibfnamefont {M.}~\bibnamefont
			{Ho}}, \bibinfo {author} {\bibfnamefont {P.}~\bibnamefont {Sekatski}},
		\bibinfo {author} {\bibfnamefont {E.~Y.-Z.}\ \bibnamefont {Tan}}, \bibinfo
		{author} {\bibfnamefont {R.}~\bibnamefont {Renner}}, \bibinfo {author}
		{\bibfnamefont {J.-D.}\ \bibnamefont {Bancal}},\ and\ \bibinfo {author}
		{\bibfnamefont {N.}~\bibnamefont {Sangouard}},\ }\bibfield  {title} {\bibinfo
		{title} {Noisy preprocessing facilitates a photonic realization of
			device-independent quantum key distribution},\ }\href
	{https://doi.org/10.1103/PhysRevLett.124.230502} {\bibfield  {journal}
		{\bibinfo  {journal} {Phys. Rev. Lett.}\ }\textbf {\bibinfo {volume} {124}},\
		\bibinfo {pages} {230502} (\bibinfo {year} {2020})}\BibitemShut {NoStop}%
	\bibitem [{\citenamefont {Zapatero}\ \emph {et~al.}(2023)\citenamefont
		{Zapatero}, \citenamefont {van Leent}, \citenamefont {Arnon-Friedman},
		\citenamefont {Liu}, \citenamefont {Zhang}, \citenamefont {Weinfurter},\ and\
		\citenamefont {Curty}}]{Zapatero2023}%
	\BibitemOpen
	\bibfield  {author} {\bibinfo {author} {\bibfnamefont {V.}~\bibnamefont
			{Zapatero}}, \bibinfo {author} {\bibfnamefont {T.}~\bibnamefont {van Leent}},
		\bibinfo {author} {\bibfnamefont {R.}~\bibnamefont {Arnon-Friedman}},
		\bibinfo {author} {\bibfnamefont {W.-Z.}\ \bibnamefont {Liu}}, \bibinfo
		{author} {\bibfnamefont {Q.}~\bibnamefont {Zhang}}, \bibinfo {author}
		{\bibfnamefont {H.}~\bibnamefont {Weinfurter}},\ and\ \bibinfo {author}
		{\bibfnamefont {M.}~\bibnamefont {Curty}},\ }\bibfield  {title} {\bibinfo
		{title} {Advances in device-independent quantum key distribution},\ }\href
	{https://doi.org/https://doi.org/10.1038/s41534-023-00684-x} {\bibfield
		{journal} {\bibinfo  {journal} {npj quantum information}\ }\textbf {\bibinfo
			{volume} {9}},\ \bibinfo {pages} {10} (\bibinfo {year} {2023})}\BibitemShut
	{NoStop}%
	\bibitem [{\citenamefont {Paw\l{}owski}\ and\ \citenamefont
		{Brunner}(2011)}]{PawlowskiSemiDIQKD2011}%
	\BibitemOpen
	\bibfield  {author} {\bibinfo {author} {\bibfnamefont {M.}~\bibnamefont
			{Paw\l{}owski}}\ and\ \bibinfo {author} {\bibfnamefont {N.}~\bibnamefont
			{Brunner}},\ }\bibfield  {title} {\bibinfo {title} {Semi-device-independent
			security of one-way quantum key distribution},\ }\href
	{https://doi.org/10.1103/PhysRevA.84.010302} {\bibfield  {journal} {\bibinfo
			{journal} {Phys. Rev. A}\ }\textbf {\bibinfo {volume} {84}},\ \bibinfo
		{pages} {010302} (\bibinfo {year} {2011})}\BibitemShut {NoStop}%
	\bibitem [{\citenamefont {Gupta}\ \emph {et~al.}(2023)\citenamefont {Gupta},
		\citenamefont {Saha}, \citenamefont {Xu}, \citenamefont {Cabello},\ and\
		\citenamefont {Majumdar}}]{GuptaSemiDiQKD2023}%
	\BibitemOpen
	\bibfield  {author} {\bibinfo {author} {\bibfnamefont {S.}~\bibnamefont
			{Gupta}}, \bibinfo {author} {\bibfnamefont {D.}~\bibnamefont {Saha}},
		\bibinfo {author} {\bibfnamefont {Z.-P.}\ \bibnamefont {Xu}}, \bibinfo
		{author} {\bibfnamefont {A.}~\bibnamefont {Cabello}},\ and\ \bibinfo {author}
		{\bibfnamefont {A.~S.}\ \bibnamefont {Majumdar}},\ }\bibfield  {title}
	{\bibinfo {title} {Quantum contextuality provides communication complexity
			advantage},\ }\href {https://doi.org/10.1103/PhysRevLett.130.080802}
	{\bibfield  {journal} {\bibinfo  {journal} {Phys. Rev. Lett.}\ }\textbf
		{\bibinfo {volume} {130}},\ \bibinfo {pages} {080802} (\bibinfo {year}
		{2023})}\BibitemShut {NoStop}%
	\bibitem [{\citenamefont {Lo}\ \emph {et~al.}(2012)\citenamefont {Lo},
		\citenamefont {Curty},\ and\ \citenamefont {Qi}}]{Lo2012MDIQKD}%
	\BibitemOpen
	\bibfield  {author} {\bibinfo {author} {\bibfnamefont {H.-K.}\ \bibnamefont
			{Lo}}, \bibinfo {author} {\bibfnamefont {M.}~\bibnamefont {Curty}},\ and\
		\bibinfo {author} {\bibfnamefont {B.}~\bibnamefont {Qi}},\ }\bibfield
	{title} {\bibinfo {title} {Measurement-device-independent quantum key
			distribution},\ }\href {https://doi.org/10.1103/PhysRevLett.108.130503}
	{\bibfield  {journal} {\bibinfo  {journal} {Phys. Rev. Lett.}\ }\textbf
		{\bibinfo {volume} {108}},\ \bibinfo {pages} {130503} (\bibinfo {year}
		{2012})}\BibitemShut {NoStop}%
	\bibitem [{\citenamefont {Choi}\ \emph {et~al.}(2016)\citenamefont {Choi},
		\citenamefont {Kwon}, \citenamefont {Woo}, \citenamefont {Oh}, \citenamefont
		{Han}, \citenamefont {Kim},\ and\ \citenamefont {Moon}}]{YujunMDI2016}%
	\BibitemOpen
	\bibfield  {author} {\bibinfo {author} {\bibfnamefont {Y.}~\bibnamefont
			{Choi}}, \bibinfo {author} {\bibfnamefont {O.}~\bibnamefont {Kwon}}, \bibinfo
		{author} {\bibfnamefont {M.}~\bibnamefont {Woo}}, \bibinfo {author}
		{\bibfnamefont {K.}~\bibnamefont {Oh}}, \bibinfo {author} {\bibfnamefont
			{S.-W.}\ \bibnamefont {Han}}, \bibinfo {author} {\bibfnamefont {Y.-S.}\
			\bibnamefont {Kim}},\ and\ \bibinfo {author} {\bibfnamefont {S.}~\bibnamefont
			{Moon}},\ }\bibfield  {title} {\bibinfo {title} {Plug-and-play
			measurement-device-independent quantum key distribution},\ }\href
	{https://doi.org/10.1103/PhysRevA.93.032319} {\bibfield  {journal} {\bibinfo
			{journal} {Phys. Rev. A}\ }\textbf {\bibinfo {volume} {93}},\ \bibinfo
		{pages} {032319} (\bibinfo {year} {2016})}\BibitemShut {NoStop}%
	\bibitem [{\citenamefont {Schrödinger}(1935)}]{Schrödinger1935}%
	\BibitemOpen
	\bibfield  {author} {\bibinfo {author} {\bibfnamefont {E.}~\bibnamefont
			{Schrödinger}},\ }\bibfield  {title} {\bibinfo {title} {Discussion of
			probability relations between separated systems},\ }\href
	{https://doi.org/10.1017/S0305004100013554} {\bibfield  {journal} {\bibinfo
			{journal} {Mathematical Proceedings of the Cambridge Philosophical Society}\
		}\textbf {\bibinfo {volume} {31}},\ \bibinfo {pages} {555–563} (\bibinfo
		{year} {1935})}\BibitemShut {NoStop}%
	\bibitem [{\citenamefont {Wiseman}\ \emph {et~al.}(2007)\citenamefont
		{Wiseman}, \citenamefont {Jones},\ and\ \citenamefont
		{Doherty}}]{WisemanSteering2007}%
	\BibitemOpen
	\bibfield  {author} {\bibinfo {author} {\bibfnamefont {H.~M.}\ \bibnamefont
			{Wiseman}}, \bibinfo {author} {\bibfnamefont {S.~J.}\ \bibnamefont {Jones}},\
		and\ \bibinfo {author} {\bibfnamefont {A.~C.}\ \bibnamefont {Doherty}},\
	}\bibfield  {title} {\bibinfo {title} {Steering, entanglement, nonlocality,
			and the einstein-podolsky-rosen paradox},\ }\href
	{https://doi.org/10.1103/PhysRevLett.98.140402} {\bibfield  {journal}
		{\bibinfo  {journal} {Phys. Rev. Lett.}\ }\textbf {\bibinfo {volume} {98}},\
		\bibinfo {pages} {140402} (\bibinfo {year} {2007})}\BibitemShut {NoStop}%
	\bibitem [{\citenamefont {Reid}(1989)}]{Reid1989}%
	\BibitemOpen
	\bibfield  {author} {\bibinfo {author} {\bibfnamefont {M.~D.}\ \bibnamefont
			{Reid}},\ }\bibfield  {title} {\bibinfo {title} {Demonstration of the
			einstein-podolsky-rosen paradox using nondegenerate parametric
			amplification},\ }\href {https://doi.org/10.1103/PhysRevA.40.913} {\bibfield
		{journal} {\bibinfo  {journal} {Phys. Rev. A}\ }\textbf {\bibinfo {volume}
			{40}},\ \bibinfo {pages} {913} (\bibinfo {year} {1989})}\BibitemShut
	{NoStop}%
	\bibitem [{\citenamefont {Cavalcanti}\ \emph {et~al.}(2009)\citenamefont
		{Cavalcanti}, \citenamefont {Jones}, \citenamefont {Wiseman},\ and\
		\citenamefont {Reid}}]{CavalcantiExptCriteriaSteering2009}%
	\BibitemOpen
	\bibfield  {author} {\bibinfo {author} {\bibfnamefont {E.~G.}\ \bibnamefont
			{Cavalcanti}}, \bibinfo {author} {\bibfnamefont {S.~J.}\ \bibnamefont
			{Jones}}, \bibinfo {author} {\bibfnamefont {H.~M.}\ \bibnamefont {Wiseman}},\
		and\ \bibinfo {author} {\bibfnamefont {M.~D.}\ \bibnamefont {Reid}},\
	}\bibfield  {title} {\bibinfo {title} {Experimental criteria for steering and
			the einstein-podolsky-rosen paradox},\ }\href
	{https://doi.org/10.1103/PhysRevA.80.032112} {\bibfield  {journal} {\bibinfo
			{journal} {Phys. Rev. A}\ }\textbf {\bibinfo {volume} {80}},\ \bibinfo
		{pages} {032112} (\bibinfo {year} {2009})}\BibitemShut {NoStop}%
	\bibitem [{\citenamefont {Chen}\ \emph {et~al.}(2013)\citenamefont {Chen},
		\citenamefont {Wu}, \citenamefont {Kwek}, \citenamefont {Oh},\ and\
		\citenamefont {Ge}}]{ChenCJWR2013}%
	\BibitemOpen
	\bibfield  {author} {\bibinfo {author} {\bibfnamefont {J.-L.}\ \bibnamefont
			{Chen}}, \bibinfo {author} {\bibfnamefont {C.}~\bibnamefont {Wu}}, \bibinfo
		{author} {\bibfnamefont {L.-C.}\ \bibnamefont {Kwek}}, \bibinfo {author}
		{\bibfnamefont {C.~H.}\ \bibnamefont {Oh}},\ and\ \bibinfo {author}
		{\bibfnamefont {M.-L.}\ \bibnamefont {Ge}},\ }\bibfield  {title} {\bibinfo
		{title} {A universal steering criterion},\ }\href
	{https://doi.org/10.1038/srep02143} {\bibfield  {journal} {\bibinfo
			{journal} {Scientific Reports}\ }\textbf {\bibinfo {volume} {3}},\ \bibinfo
		{pages} {2143} (\bibinfo {year} {2013})}\BibitemShut {NoStop}%
	\bibitem [{\citenamefont {Schneeloch}\ \emph {et~al.}(2013)\citenamefont
		{Schneeloch}, \citenamefont {Broadbent}, \citenamefont {Walborn},
		\citenamefont {Cavalcanti},\ and\ \citenamefont {Howell}}]{Schneeloch2013}%
	\BibitemOpen
	\bibfield  {author} {\bibinfo {author} {\bibfnamefont {J.}~\bibnamefont
			{Schneeloch}}, \bibinfo {author} {\bibfnamefont {C.~J.}\ \bibnamefont
			{Broadbent}}, \bibinfo {author} {\bibfnamefont {S.~P.}\ \bibnamefont
			{Walborn}}, \bibinfo {author} {\bibfnamefont {E.~G.}\ \bibnamefont
			{Cavalcanti}},\ and\ \bibinfo {author} {\bibfnamefont {J.~C.}\ \bibnamefont
			{Howell}},\ }\bibfield  {title} {\bibinfo {title} {Einstein-podolsky-rosen
			steering inequalities from entropic uncertainty relations},\ }\href
	{https://doi.org/10.1103/PhysRevA.87.062103} {\bibfield  {journal} {\bibinfo
			{journal} {Phys. Rev. A}\ }\textbf {\bibinfo {volume} {87}},\ \bibinfo
		{pages} {062103} (\bibinfo {year} {2013})}\BibitemShut {NoStop}%
	\bibitem [{\citenamefont {Pramanik}\ \emph {et~al.}(2014)\citenamefont
		{Pramanik}, \citenamefont {Kaplan},\ and\ \citenamefont
		{Majumdar}}]{Pramanik2014}%
	\BibitemOpen
	\bibfield  {author} {\bibinfo {author} {\bibfnamefont {T.}~\bibnamefont
			{Pramanik}}, \bibinfo {author} {\bibfnamefont {M.}~\bibnamefont {Kaplan}},\
		and\ \bibinfo {author} {\bibfnamefont {A.~S.}\ \bibnamefont {Majumdar}},\
	}\bibfield  {title} {\bibinfo {title} {Fine-grained
			einstein-podolsky-rosen--steering inequalities},\ }\href
	{https://doi.org/10.1103/PhysRevA.90.050305} {\bibfield  {journal} {\bibinfo
			{journal} {Phys. Rev. A}\ }\textbf {\bibinfo {volume} {90}},\ \bibinfo
		{pages} {050305} (\bibinfo {year} {2014})}\BibitemShut {NoStop}%
	\bibitem [{\citenamefont {Maity}\ \emph {et~al.}(2017)\citenamefont {Maity},
		\citenamefont {Datta},\ and\ \citenamefont {Majumdar}}]{sumuncert}%
	\BibitemOpen
	\bibfield  {author} {\bibinfo {author} {\bibfnamefont {A.~G.}\ \bibnamefont
			{Maity}}, \bibinfo {author} {\bibfnamefont {S.}~\bibnamefont {Datta}},\ and\
		\bibinfo {author} {\bibfnamefont {A.~S.}\ \bibnamefont {Majumdar}},\
	}\bibfield  {title} {\bibinfo {title} {Tighter einstein-podolsky-rosen
			steering inequality based on the sum steering relation},\ }\href
	{https://doi.org/https://doi.org/10.1103/PhysRevA.96.052326} {\bibfield
		{journal} {\bibinfo  {journal} {Phys. Rev. A}\ }\textbf {\bibinfo {volume}
			{96}},\ \bibinfo {pages} {052326} (\bibinfo {year} {2017})}\BibitemShut
	{NoStop}%
	\bibitem [{\citenamefont {Jones}\ \emph {et~al.}(2007)\citenamefont {Jones},
		\citenamefont {Wiseman},\ and\ \citenamefont {Doherty}}]{Jones2007Steering}%
	\BibitemOpen
	\bibfield  {author} {\bibinfo {author} {\bibfnamefont {S.~J.}\ \bibnamefont
			{Jones}}, \bibinfo {author} {\bibfnamefont {H.~M.}\ \bibnamefont {Wiseman}},\
		and\ \bibinfo {author} {\bibfnamefont {A.~C.}\ \bibnamefont {Doherty}},\
	}\bibfield  {title} {\bibinfo {title} {Entanglement, einstein-podolsky-rosen
			correlations, bell nonlocality, and steering},\ }\href
	{https://doi.org/10.1103/PhysRevA.76.052116} {\bibfield  {journal} {\bibinfo
			{journal} {Phys. Rev. A}\ }\textbf {\bibinfo {volume} {76}},\ \bibinfo
		{pages} {052116} (\bibinfo {year} {2007})}\BibitemShut {NoStop}%
	\bibitem [{\citenamefont {Walborn}\ \emph {et~al.}(2011)\citenamefont
		{Walborn}, \citenamefont {Salles}, \citenamefont {Gomes}, \citenamefont
		{Toscano},\ and\ \citenamefont {Souto~Ribeiro}}]{Walborn2011}%
	\BibitemOpen
	\bibfield  {author} {\bibinfo {author} {\bibfnamefont {S.~P.}\ \bibnamefont
			{Walborn}}, \bibinfo {author} {\bibfnamefont {A.}~\bibnamefont {Salles}},
		\bibinfo {author} {\bibfnamefont {R.~M.}\ \bibnamefont {Gomes}}, \bibinfo
		{author} {\bibfnamefont {F.}~\bibnamefont {Toscano}},\ and\ \bibinfo {author}
		{\bibfnamefont {P.~H.}\ \bibnamefont {Souto~Ribeiro}},\ }\bibfield  {title}
	{\bibinfo {title} {Revealing hidden einstein-podolsky-rosen nonlocality},\
	}\href {https://doi.org/10.1103/PhysRevLett.106.130402} {\bibfield  {journal}
		{\bibinfo  {journal} {Phys. Rev. Lett.}\ }\textbf {\bibinfo {volume} {106}},\
		\bibinfo {pages} {130402} (\bibinfo {year} {2011})}\BibitemShut {NoStop}%
	\bibitem [{\citenamefont {Saunders}\ \emph {et~al.}(2010)\citenamefont
		{Saunders}, \citenamefont {Jones}, \citenamefont {Wiseman},\ and\
		\citenamefont {Pryde}}]{Saunders2010EPRSteering}%
	\BibitemOpen
	\bibfield  {author} {\bibinfo {author} {\bibfnamefont {D.~J.}\ \bibnamefont
			{Saunders}}, \bibinfo {author} {\bibfnamefont {S.~J.}\ \bibnamefont {Jones}},
		\bibinfo {author} {\bibfnamefont {H.~M.}\ \bibnamefont {Wiseman}},\ and\
		\bibinfo {author} {\bibfnamefont {G.~J.}\ \bibnamefont {Pryde}},\ }\bibfield
	{title} {\bibinfo {title} {Experimental {EPR}-steering using {Bell}-local
			states},\ }\href {https://doi.org/10.1038/nphys1766} {\bibfield  {journal}
		{\bibinfo  {journal} {Nature Physics}\ }\textbf {\bibinfo {volume} {6}},\
		\bibinfo {pages} {845} (\bibinfo {year} {2010})}\BibitemShut {NoStop}%
	\bibitem [{\citenamefont {Costa}\ and\ \citenamefont
		{Angelo}(2016)}]{CostaQuantificationSteering2016}%
	\BibitemOpen
	\bibfield  {author} {\bibinfo {author} {\bibfnamefont {A.~C.~S.}\
			\bibnamefont {Costa}}\ and\ \bibinfo {author} {\bibfnamefont {R.~M.}\
			\bibnamefont {Angelo}},\ }\bibfield  {title} {\bibinfo {title}
		{Quantification of einstein-podolsky-rosen steering for two-qubit states},\
	}\href {https://doi.org/10.1103/PhysRevA.93.020103} {\bibfield  {journal}
		{\bibinfo  {journal} {Phys. Rev. A}\ }\textbf {\bibinfo {volume} {93}},\
		\bibinfo {pages} {020103} (\bibinfo {year} {2016})}\BibitemShut {NoStop}%
	\bibitem [{\citenamefont {Uola}\ \emph {et~al.}(2020)\citenamefont {Uola},
		\citenamefont {Costa}, \citenamefont {Nguyen},\ and\ \citenamefont
		{G\"uhne}}]{UolaQuantumsteering2020}%
	\BibitemOpen
	\bibfield  {author} {\bibinfo {author} {\bibfnamefont {R.}~\bibnamefont
			{Uola}}, \bibinfo {author} {\bibfnamefont {A.~C.~S.}\ \bibnamefont {Costa}},
		\bibinfo {author} {\bibfnamefont {H.~C.}\ \bibnamefont {Nguyen}},\ and\
		\bibinfo {author} {\bibfnamefont {O.}~\bibnamefont {G\"uhne}},\ }\bibfield
	{title} {\bibinfo {title} {Quantum steering},\ }\href
	{https://doi.org/10.1103/RevModPhys.92.015001} {\bibfield  {journal}
		{\bibinfo  {journal} {Rev. Mod. Phys.}\ }\textbf {\bibinfo {volume} {92}},\
		\bibinfo {pages} {015001} (\bibinfo {year} {2020})}\BibitemShut {NoStop}%
	\bibitem [{\citenamefont {Cavalcanti}\ and\ \citenamefont
		{Skrzypczyk}(2016)}]{CavalcantiSteeringReview2017}%
	\BibitemOpen
	\bibfield  {author} {\bibinfo {author} {\bibfnamefont {D.}~\bibnamefont
			{Cavalcanti}}\ and\ \bibinfo {author} {\bibfnamefont {P.}~\bibnamefont
			{Skrzypczyk}},\ }\bibfield  {title} {\bibinfo {title} {Quantum steering: a
			review with focus on semidefinite programming},\ }\href
	{https://doi.org/10.1088/1361-6633/80/2/024001} {\bibfield  {journal}
		{\bibinfo  {journal} {Reports on Progress in Physics}\ }\textbf {\bibinfo
			{volume} {80}},\ \bibinfo {pages} {024001} (\bibinfo {year}
		{2016})}\BibitemShut {NoStop}%
	\bibitem [{\citenamefont {Das}\ \emph {et~al.}(2018)\citenamefont {Das},
		\citenamefont {Datta}, \citenamefont {Jebaratnam},\ and\ \citenamefont
		{Majumdar}}]{JebaSteeringCost}%
	\BibitemOpen
	\bibfield  {author} {\bibinfo {author} {\bibfnamefont {D.}~\bibnamefont
			{Das}}, \bibinfo {author} {\bibfnamefont {S.}~\bibnamefont {Datta}}, \bibinfo
		{author} {\bibfnamefont {C.}~\bibnamefont {Jebaratnam}},\ and\ \bibinfo
		{author} {\bibfnamefont {A.~S.}\ \bibnamefont {Majumdar}},\ }\href
	{https://doi.org/10.1103/PhysRevA.97.022110} {\bibfield  {journal} {\bibinfo
			{journal} {Phys. Rev. A}\ }\textbf {\bibinfo {volume} {97}},\ \bibinfo
		{pages} {022110} (\bibinfo {year} {2018})}\BibitemShut {NoStop}%
	\bibitem [{\citenamefont {Branciard}\ \emph {et~al.}(2012)\citenamefont
		{Branciard}, \citenamefont {Cavalcanti}, \citenamefont {Walborn},
		\citenamefont {Scarani},\ and\ \citenamefont {Wiseman}}]{Branciard2012}%
	\BibitemOpen
	\bibfield  {author} {\bibinfo {author} {\bibfnamefont {C.}~\bibnamefont
			{Branciard}}, \bibinfo {author} {\bibfnamefont {E.~G.}\ \bibnamefont
			{Cavalcanti}}, \bibinfo {author} {\bibfnamefont {S.~P.}\ \bibnamefont
			{Walborn}}, \bibinfo {author} {\bibfnamefont {V.}~\bibnamefont {Scarani}},\
		and\ \bibinfo {author} {\bibfnamefont {H.~M.}\ \bibnamefont {Wiseman}},\
	}\bibfield  {title} {\bibinfo {title} {One-sided device-independent quantum
			key distribution: Security, feasibility, and the connection with steering},\
	}\href {https://doi.org/10.1103/PhysRevA.85.010301} {\bibfield  {journal}
		{\bibinfo  {journal} {Phys. Rev. A}\ }\textbf {\bibinfo {volume} {85}},\
		\bibinfo {pages} {010301} (\bibinfo {year} {2012})}\BibitemShut {NoStop}%
	\bibitem [{\citenamefont {Tomamichel}\ \emph {et~al.}(2013)\citenamefont
		{Tomamichel}, \citenamefont {Fehr}, \citenamefont {Kaniewski},\ and\
		\citenamefont {Wehner}}]{Tomamichel1sDIQKD2013}%
	\BibitemOpen
	\bibfield  {author} {\bibinfo {author} {\bibfnamefont {M.}~\bibnamefont
			{Tomamichel}}, \bibinfo {author} {\bibfnamefont {S.}~\bibnamefont {Fehr}},
		\bibinfo {author} {\bibfnamefont {J.}~\bibnamefont {Kaniewski}},\ and\
		\bibinfo {author} {\bibfnamefont {S.}~\bibnamefont {Wehner}},\ }\bibfield
	{title} {\bibinfo {title} {One-sided device-independent qkd and
			position-based cryptography from monogamy games},\ }in\ \href@noop {} {\emph
		{\bibinfo {booktitle} {Advances in Cryptology -- EUROCRYPT 2013}}},\ \bibinfo
	{editor} {edited by\ \bibinfo {editor} {\bibfnamefont {T.}~\bibnamefont
			{Johansson}}\ and\ \bibinfo {editor} {\bibfnamefont {P.~Q.}\ \bibnamefont
			{Nguyen}}}\ (\bibinfo  {publisher} {Springer Berlin Heidelberg},\ \bibinfo
	{address} {Berlin, Heidelberg},\ \bibinfo {year} {2013})\ pp.\ \bibinfo
	{pages} {609--625}\BibitemShut {NoStop}%
	\bibitem [{\citenamefont {Mukherjee}\ \emph {et~al.}(2023)\citenamefont
		{Mukherjee}, \citenamefont {Patro},\ and\ \citenamefont
		{Ganguly}}]{Mukherjee2023SteeringQKD}%
	\BibitemOpen
	\bibfield  {author} {\bibinfo {author} {\bibfnamefont {K.}~\bibnamefont
			{Mukherjee}}, \bibinfo {author} {\bibfnamefont {T.}~\bibnamefont {Patro}},\
		and\ \bibinfo {author} {\bibfnamefont {N.}~\bibnamefont {Ganguly}},\
	}\bibfield  {title} {\bibinfo {title} {Role of steering inequality in quantum
			key distribution protocol},\ }\href
	{https://doi.org/10.12743/quanta.v12i1.210} {\bibfield  {journal} {\bibinfo
			{journal} {Quanta}\ }\textbf {\bibinfo {volume} {12}} (\bibinfo {year}
		{2023})},\ \bibinfo {note} {published: 2023-04-18}\BibitemShut {NoStop}%
	\bibitem [{\citenamefont {Masini}\ and\ \citenamefont
		{Sarkar}(2024)}]{masini20241sDIQKD}%
	\BibitemOpen
	\bibfield  {author} {\bibinfo {author} {\bibfnamefont {M.}~\bibnamefont
			{Masini}}\ and\ \bibinfo {author} {\bibfnamefont {S.}~\bibnamefont
			{Sarkar}},\ }\bibfield  {title} {\bibinfo {title} {One-sided di-qkd secure
			against coherent attacks over long distances},\ }\href
	{https://arxiv.org/abs/2403.11850} {\bibfield  {journal} {\bibinfo  {journal}
			{arXiv preprint arXiv:2403.11850}\ } (\bibinfo {year} {2024})}\BibitemShut
	{NoStop}%
	\bibitem [{\citenamefont {Bennet}\ \emph {et~al.}(2012)\citenamefont {Bennet},
		\citenamefont {Evans}, \citenamefont {Saunders}, \citenamefont {Branciard},
		\citenamefont {Cavalcanti}, \citenamefont {Wiseman},\ and\ \citenamefont
		{Pryde}}]{BennetExptSteering2012}%
	\BibitemOpen
	\bibfield  {author} {\bibinfo {author} {\bibfnamefont {A.~J.}\ \bibnamefont
			{Bennet}}, \bibinfo {author} {\bibfnamefont {D.~A.}\ \bibnamefont {Evans}},
		\bibinfo {author} {\bibfnamefont {D.~J.}\ \bibnamefont {Saunders}}, \bibinfo
		{author} {\bibfnamefont {C.}~\bibnamefont {Branciard}}, \bibinfo {author}
		{\bibfnamefont {E.~G.}\ \bibnamefont {Cavalcanti}}, \bibinfo {author}
		{\bibfnamefont {H.~M.}\ \bibnamefont {Wiseman}},\ and\ \bibinfo {author}
		{\bibfnamefont {G.~J.}\ \bibnamefont {Pryde}},\ }\bibfield  {title} {\bibinfo
		{title} {Arbitrarily loss-tolerant einstein-podolsky-rosen steering allowing
			a demonstration over 1 km of optical fiber with no detection loophole},\
	}\href {https://doi.org/10.1103/PhysRevX.2.031003} {\bibfield  {journal}
		{\bibinfo  {journal} {Phys. Rev. X}\ }\textbf {\bibinfo {volume} {2}},\
		\bibinfo {pages} {031003} (\bibinfo {year} {2012})}\BibitemShut {NoStop}%
	\bibitem [{\citenamefont {Renner}(2008)}]{Renner2008}%
	\BibitemOpen
	\bibfield  {author} {\bibinfo {author} {\bibfnamefont {R.}~\bibnamefont
			{Renner}},\ }\bibfield  {title} {\bibinfo {title} {Security of quantum key
			distribution},\ }\href {https://doi.org/10.1142/S0219749908003256} {\bibfield
		{journal} {\bibinfo  {journal} {International Journal of Quantum
				Information}\ }\textbf {\bibinfo {volume} {6}},\ \bibinfo {pages} {1}
		(\bibinfo {year} {2008})}\BibitemShut {NoStop}%
	\bibitem [{\citenamefont {Devetak}\ and\ \citenamefont
		{Winter}(2005)}]{Devetak2005}%
	\BibitemOpen
	\bibfield  {author} {\bibinfo {author} {\bibfnamefont {I.}~\bibnamefont
			{Devetak}}\ and\ \bibinfo {author} {\bibfnamefont {A.}~\bibnamefont
			{Winter}},\ }\bibfield  {title} {\bibinfo {title} {Distillation of secret key
			and entanglement from quantum states},\ }\href
	{https://doi.org/10.1098/rspa.2004.1372} {\bibfield  {journal} {\bibinfo
			{journal} {Proceedings of the Royal Society A: Mathematical, Physical and
				Engineering Sciences}\ }\textbf {\bibinfo {volume} {461}},\ \bibinfo {pages}
		{207–235} (\bibinfo {year} {2005})}\BibitemShut {NoStop}%
	\bibitem [{\citenamefont {Acín}\ \emph {et~al.}(2006)\citenamefont {Acín},
		\citenamefont {Massar},\ and\ \citenamefont {Pironio}}]{Acin2006njp}%
	\BibitemOpen
	\bibfield  {author} {\bibinfo {author} {\bibfnamefont {A.}~\bibnamefont
			{Acín}}, \bibinfo {author} {\bibfnamefont {S.}~\bibnamefont {Massar}},\ and\
		\bibinfo {author} {\bibfnamefont {S.}~\bibnamefont {Pironio}},\ }\bibfield
	{title} {\bibinfo {title} {Efficient quantum key distribution secure against
			no-signalling eavesdroppers},\ }\href
	{https://doi.org/10.1088/1367-2630/8/8/126} {\bibfield  {journal} {\bibinfo
			{journal} {New Journal of Physics}\ }\textbf {\bibinfo {volume} {8}},\
		\bibinfo {pages} {126–126} (\bibinfo {year} {2006})}\BibitemShut {NoStop}%
	\bibitem [{\citenamefont {Barrett}\ \emph {et~al.}(2005)\citenamefont
		{Barrett}, \citenamefont {Hardy},\ and\ \citenamefont {Kent}}]{Barrett2005a}%
	\BibitemOpen
	\bibfield  {author} {\bibinfo {author} {\bibfnamefont {J.}~\bibnamefont
			{Barrett}}, \bibinfo {author} {\bibfnamefont {L.}~\bibnamefont {Hardy}},\
		and\ \bibinfo {author} {\bibfnamefont {A.}~\bibnamefont {Kent}},\ }\bibfield
	{title} {\bibinfo {title} {No signaling and quantum key distribution},\
	}\href {https://doi.org/10.1103/PhysRevLett.95.010503} {\bibfield  {journal}
		{\bibinfo  {journal} {Phys. Rev. Lett.}\ }\textbf {\bibinfo {volume} {95}},\
		\bibinfo {pages} {010503} (\bibinfo {year} {2005})}\BibitemShut {NoStop}%
	\bibitem [{\citenamefont {Konig}\ \emph {et~al.}(2009)\citenamefont {Konig},
		\citenamefont {Renner},\ and\ \citenamefont
		{Schaffner}}]{KonigMinMaxEntropy2009}%
	\BibitemOpen
	\bibfield  {author} {\bibinfo {author} {\bibfnamefont {R.}~\bibnamefont
			{Konig}}, \bibinfo {author} {\bibfnamefont {R.}~\bibnamefont {Renner}},\ and\
		\bibinfo {author} {\bibfnamefont {C.}~\bibnamefont {Schaffner}},\ }\bibfield
	{title} {\bibinfo {title} {The operational meaning of min- and max-entropy},\
	}\href {https://doi.org/10.1109/TIT.2009.2025545} {\bibfield  {journal}
		{\bibinfo  {journal} {IEEE Transactions on Information Theory}\ }\textbf
		{\bibinfo {volume} {55}},\ \bibinfo {pages} {4337} (\bibinfo {year}
		{2009})}\BibitemShut {NoStop}%
	\bibitem [{\citenamefont {Tomamichel}\ \emph {et~al.}(2012)\citenamefont
		{Tomamichel}, \citenamefont {Lim}, \citenamefont {Gisin},\ and\ \citenamefont
		{Renner}}]{Tomamichel2012FiniteKey}%
	\BibitemOpen
	\bibfield  {author} {\bibinfo {author} {\bibfnamefont {M.}~\bibnamefont
			{Tomamichel}}, \bibinfo {author} {\bibfnamefont {C.~C.~W.}\ \bibnamefont
			{Lim}}, \bibinfo {author} {\bibfnamefont {N.}~\bibnamefont {Gisin}},\ and\
		\bibinfo {author} {\bibfnamefont {R.}~\bibnamefont {Renner}},\ }\bibfield
	{title} {\bibinfo {title} {Tight finite-key analysis for quantum
			cryptography},\ }\href {https://doi.org/10.1038/ncomms1631} {\bibfield
		{journal} {\bibinfo  {journal} {Nature Communications}\ }\textbf {\bibinfo
			{volume} {3}},\ \bibinfo {pages} {634} (\bibinfo {year} {2012})}\BibitemShut
	{NoStop}%
	\bibitem [{\citenamefont {Arnon-Friedman}\ \emph {et~al.}(2018)\citenamefont
		{Arnon-Friedman}, \citenamefont {Dupuis}, \citenamefont {Fawzi},
		\citenamefont {Renner},\ and\ \citenamefont {Vidick}}]{FriedmanEAT2018}%
	\BibitemOpen
	\bibfield  {author} {\bibinfo {author} {\bibfnamefont {R.}~\bibnamefont
			{Arnon-Friedman}}, \bibinfo {author} {\bibfnamefont {F.}~\bibnamefont
			{Dupuis}}, \bibinfo {author} {\bibfnamefont {O.}~\bibnamefont {Fawzi}},
		\bibinfo {author} {\bibfnamefont {R.}~\bibnamefont {Renner}},\ and\ \bibinfo
		{author} {\bibfnamefont {T.}~\bibnamefont {Vidick}},\ }\bibfield  {title}
	{\bibinfo {title} {Practical device-independent quantum cryptography via
			entropy accumulation},\ }\href {https://doi.org/10.1038/s41467-017-02307-4}
	{\bibfield  {journal} {\bibinfo  {journal} {Nature Communications}\ }\textbf
		{\bibinfo {volume} {9}},\ \bibinfo {pages} {459} (\bibinfo {year}
		{2018})}\BibitemShut {NoStop}%
	\bibitem [{\citenamefont {Dupuis}\ \emph {et~al.}(2020)\citenamefont {Dupuis},
		\citenamefont {Fawzi},\ and\ \citenamefont {Renner}}]{Dupuis2020EAT}%
	\BibitemOpen
	\bibfield  {author} {\bibinfo {author} {\bibfnamefont {F.}~\bibnamefont
			{Dupuis}}, \bibinfo {author} {\bibfnamefont {O.}~\bibnamefont {Fawzi}},\ and\
		\bibinfo {author} {\bibfnamefont {R.}~\bibnamefont {Renner}},\ }\bibfield
	{title} {\bibinfo {title} {Entropy accumulation},\ }\href
	{https://doi.org/10.1007/s00220-020-03839-5} {\bibfield  {journal} {\bibinfo
			{journal} {Communications in Mathematical Physics}\ }\textbf {\bibinfo
			{volume} {379}},\ \bibinfo {pages} {867} (\bibinfo {year}
		{2020})}\BibitemShut {NoStop}%
	\bibitem [{\citenamefont {Masanes}(2006)}]{Masanes2006}%
	\BibitemOpen
	\bibfield  {author} {\bibinfo {author} {\bibfnamefont {L.}~\bibnamefont
			{Masanes}},\ }\bibfield  {title} {\bibinfo {title} {Asymptotic violation of
			bell inequalities and distillability},\ }\href
	{https://doi.org/10.1103/PhysRevLett.97.050503} {\bibfield  {journal}
		{\bibinfo  {journal} {Phys. Rev. Lett.}\ }\textbf {\bibinfo {volume} {97}},\
		\bibinfo {pages} {050503} (\bibinfo {year} {2006})}\BibitemShut {NoStop}%
	\bibitem [{\citenamefont {Lounesto}(2001)}]{Lounesto2001}%
	\BibitemOpen
	\bibfield  {author} {\bibinfo {author} {\bibfnamefont {P.}~\bibnamefont
			{Lounesto}},\ }\href@noop {} {\emph {\bibinfo {title} {Clifford Algebras and
				Spinors}}},\ \bibinfo {edition} {2nd}\ ed.\ (\bibinfo  {publisher} {Cambridge
		University Press},\ \bibinfo {year} {2001})\BibitemShut {NoStop}%
	\bibitem [{\citenamefont {Pirandola}\ \emph {et~al.}(2020)\citenamefont
		{Pirandola}, \citenamefont {Andersen}, \citenamefont {Banchi}, \citenamefont
		{Berta}, \citenamefont {Bunandar}, \citenamefont {Colbeck}, \citenamefont
		{Englund}, \citenamefont {Gehring}, \citenamefont {Lupo}, \citenamefont
		{Ottaviani}, \citenamefont {Pereira}, \citenamefont {Razavi}, \citenamefont
		{Shamsul~Shaari}, \citenamefont {Tomamichel}, \citenamefont {Usenko},
		\citenamefont {Vallone}, \citenamefont {Villoresi},\ and\ \citenamefont
		{Wallden}}]{PirandolaQKDReview2020}%
	\BibitemOpen
	\bibfield  {author} {\bibinfo {author} {\bibfnamefont {S.}~\bibnamefont
			{Pirandola}}, \bibinfo {author} {\bibfnamefont {U.~L.}\ \bibnamefont
			{Andersen}}, \bibinfo {author} {\bibfnamefont {L.}~\bibnamefont {Banchi}},
		\bibinfo {author} {\bibfnamefont {M.}~\bibnamefont {Berta}}, \bibinfo
		{author} {\bibfnamefont {D.}~\bibnamefont {Bunandar}}, \bibinfo {author}
		{\bibfnamefont {R.}~\bibnamefont {Colbeck}}, \bibinfo {author} {\bibfnamefont
			{D.}~\bibnamefont {Englund}}, \bibinfo {author} {\bibfnamefont
			{T.}~\bibnamefont {Gehring}}, \bibinfo {author} {\bibfnamefont
			{C.}~\bibnamefont {Lupo}}, \bibinfo {author} {\bibfnamefont {C.}~\bibnamefont
			{Ottaviani}}, \bibinfo {author} {\bibfnamefont {J.~L.}\ \bibnamefont
			{Pereira}}, \bibinfo {author} {\bibfnamefont {M.}~\bibnamefont {Razavi}},
		\bibinfo {author} {\bibfnamefont {J.}~\bibnamefont {Shamsul~Shaari}},
		\bibinfo {author} {\bibfnamefont {M.}~\bibnamefont {Tomamichel}}, \bibinfo
		{author} {\bibfnamefont {V.~C.}\ \bibnamefont {Usenko}}, \bibinfo {author}
		{\bibfnamefont {G.}~\bibnamefont {Vallone}}, \bibinfo {author} {\bibfnamefont
			{P.}~\bibnamefont {Villoresi}},\ and\ \bibinfo {author} {\bibfnamefont
			{P.}~\bibnamefont {Wallden}},\ }\bibfield  {title} {\bibinfo {title}
		{Advances in quantum cryptography},\ }\href
	{https://doi.org/10.1364/aop.361502} {\bibfield  {journal} {\bibinfo
			{journal} {Advances in Optics and Photonics}\ }\textbf {\bibinfo {volume}
			{12}},\ \bibinfo {pages} {1012} (\bibinfo {year} {2020})}\BibitemShut
	{NoStop}%
	\bibitem [{\citenamefont {Portmann}\ and\ \citenamefont
		{Renner}(2022)}]{PortmannQKDReview2022}%
	\BibitemOpen
	\bibfield  {author} {\bibinfo {author} {\bibfnamefont {C.}~\bibnamefont
			{Portmann}}\ and\ \bibinfo {author} {\bibfnamefont {R.}~\bibnamefont
			{Renner}},\ }\bibfield  {title} {\bibinfo {title} {Security in quantum
			cryptography},\ }\href {https://doi.org/10.1103/RevModPhys.94.025008}
	{\bibfield  {journal} {\bibinfo  {journal} {Rev. Mod. Phys.}\ }\textbf
		{\bibinfo {volume} {94}},\ \bibinfo {pages} {025008} (\bibinfo {year}
		{2022})}\BibitemShut {NoStop}%
	\bibitem [{\citenamefont {Renner}\ and\ \citenamefont
		{Wolf}(2023)}]{RennerQKDReview2023}%
	\BibitemOpen
	\bibfield  {author} {\bibinfo {author} {\bibfnamefont {R.}~\bibnamefont
			{Renner}}\ and\ \bibinfo {author} {\bibfnamefont {R.}~\bibnamefont {Wolf}},\
	}\bibfield  {title} {\bibinfo {title} {Quantum advantage in cryptography},\
	}\href {https://doi.org/10.2514/1.j062267} {\bibfield  {journal} {\bibinfo
			{journal} {AIAA Journal}\ }\textbf {\bibinfo {volume} {61}},\ \bibinfo
		{pages} {1895–1910} (\bibinfo {year} {2023})}\BibitemShut {NoStop}%
	\bibitem [{\citenamefont {Christandl}\ and\ \citenamefont
		{Toner}(2009)}]{ChristandlTonerDeFinetti2009}%
	\BibitemOpen
	\bibfield  {author} {\bibinfo {author} {\bibfnamefont {M.}~\bibnamefont
			{Christandl}}\ and\ \bibinfo {author} {\bibfnamefont {B.}~\bibnamefont
			{Toner}},\ }\bibfield  {title} {\bibinfo {title} {Finite de finetti theorem
			for conditional probability distributions describing physical theories},\
	}\href {https://doi.org/10.1063/1.3114986} {\bibfield  {journal} {\bibinfo
			{journal} {Journal of Mathematical Physics}\ }\textbf {\bibinfo {volume}
			{50}},\ \bibinfo {pages} {042104} (\bibinfo {year} {2009})}\BibitemShut
	{NoStop}%
	\bibitem [{\citenamefont {Sekatski}\ \emph {et~al.}(2021)\citenamefont
		{Sekatski}, \citenamefont {Bancal}, \citenamefont {Valcarce}, \citenamefont
		{Tan}, \citenamefont {Renner},\ and\ \citenamefont
		{Sangouard}}]{Sekatski2021}%
	\BibitemOpen
	\bibfield  {author} {\bibinfo {author} {\bibfnamefont {P.}~\bibnamefont
			{Sekatski}}, \bibinfo {author} {\bibfnamefont {J.-D.}\ \bibnamefont
			{Bancal}}, \bibinfo {author} {\bibfnamefont {X.}~\bibnamefont {Valcarce}},
		\bibinfo {author} {\bibfnamefont {E.~Y.-Z.}\ \bibnamefont {Tan}}, \bibinfo
		{author} {\bibfnamefont {R.}~\bibnamefont {Renner}},\ and\ \bibinfo {author}
		{\bibfnamefont {N.}~\bibnamefont {Sangouard}},\ }\bibfield  {title} {\bibinfo
		{title} {Device-independent quantum key distribution from generalized {CHSH}
			inequalities},\ }\href {https://doi.org/10.22331/q-2021-04-26-444} {\bibfield
		{journal} {\bibinfo  {journal} {{Quantum}}\ }\textbf {\bibinfo {volume}
			{5}},\ \bibinfo {pages} {444} (\bibinfo {year} {2021})}\BibitemShut {NoStop}%
\end{thebibliography}
\end{document}